\documentclass[journal]{IEEEtran}
\normalsize

\usepackage{amsthm}
\usepackage{graphicx}
\usepackage{verbatim}
\usepackage{amssymb,amsmath}
\usepackage{mathtools}
\usepackage{color}
\usepackage[font=small]{caption}
\usepackage{epstopdf}
\usepackage{bm}
\usepackage[nosort]{cite}
\usepackage{array,color}
\usepackage{mathtools}
\usepackage[ruled,linesnumbered]{algorithm2e}
\DontPrintSemicolon
\usepackage{algpseudocode}
\usepackage{amsmath}

\newtheorem{Theorem}{\textbf{Theorem}}
\newtheorem{Remark}{\textbf{Remark}}
\usepackage{graphics}
\usepackage{epsfig}
\usepackage[normalem]{ulem}
\usepackage{cite}
\usepackage{balance}
\usepackage[square,numbers]{natbib}
\usepackage{setspace}
\usepackage{subfigure}
\usepackage{multirow}
\usepackage{multicol}
\usepackage{lipsum}
\usepackage{cuted}

\begin{document}

\title{Rateless DeepJSCC for Broadcast Channels: a Rate-Distortion-Complexity Tradeoff}

\author{
    Zijun Qin, 
    Jingxuan Huang, 
    Zesong Fei, 
    Haichuan Ding, 
    Yulin Shao, 
    and
    Xianhao Chen

    \thanks{Z. Qin, Y. Shao and X. Chen are with the Department of Electrical and Computer Engineering, University of Hong Kong, Pok Fu Lam, Hong Kong S.A.R., China (e-mail: zjun.qin@outlook.com, ylshao@hku.hk, xchen@eee.hku.hk).}
	\thanks{J. Huang and Z. Fei are with the School of Information and Electronics, Beijing Institute of Technology, Beijing, China (e-mail: feizesong@bit.edu.cn, jxhbit@gmail.com).}
    \thanks{H. Ding is with the School of Cyberspace Science and Technology, Beijing Institute of Technology, Beijing, China (e-mail: hcding@bit.edu.cn).}

    \thanks{(\emph{Corresponding author: Xianhao Chen})}
}

\maketitle

{
\begin{abstract}
In recent years, numerous data-intensive broadcasting applications have emerged at the wireless edge, calling for a flexible tradeoff between distortion, transmission rate, and processing complexity.
While deep learning-based joint source-channel coding (DeepJSCC) has been identified as a potential solution to data-intensive communications, most of these schemes are confined to worst-case solutions, lack adaptive complexity, and are inefficient in broadcast settings. To overcome these limitations, this paper introduces nonlinear transform rateless source-channel coding (NTRSCC), a variable-length JSCC framework for broadcast channels based on rateless codes. In particular, we integrate learned source transformations with physical-layer LT codes, develop unequal protection schemes that exploit decoder side information, and devise approximations to enable end-to-end optimization of rateless parameters. Our framework enables \textit{heterogeneous} receivers to adaptively adjust their received number of rateless symbols and decoding iterations in belief propagation, thereby achieving a controllable tradeoff between distortion, rate, and decoding complexity. Simulation results demonstrate that the proposed method enhances image broadcast quality under stringent communication and processing budgets over heterogeneous edge devices.
\end{abstract}


\section{Introduction}

In emerging real-time applications at the wireless edge, environmental information, such as camera or LiDAR data, often needs to be \textit{broadcast} to heterogeneous edge devices, including vehicles, robots, and drones. One example is infrastructure-assisted autonomous driving \cite{10449899}, where a roadside unit broadcasts data-intensive sensor streams (e.g., LiDAR point clouds or camera feeds from an elevated vantage point) to nearby vehicles, allowing them to extend their perception beyond their own limited fields of view.
Although advanced methods exist to significantly compress sensory data~\cite{shi_vips_2022}, and modern channel codes can approach channel capacity at asymptotic lengths \cite{1-b7} \cite{1-b37}, the issue of efficient and adaptive broadcast over heterogeneous networks remains an open challenge.
Traditional separation-based source and channel coding works effectively only under infinite block lengths, and fails to exploit the potential correlation between data and channel. In practical scenarios, it fails to flexibly adapt compression and protection levels to different receivers and data patterns, leading to increased delay, high complexity, and poor fidelity as the channel degrades. Conversely, joint source–channel coding (JSCC) provides a principled way to balance these competing objectives. By adaptively compressing and protecting task-relevant information according to channel conditions, JSCC preserves the quality of received information even under stringent communication constraints.

Leveraging the data-driven nature of deep neural networks, a plethora of DeepJSCC methods have been proposed in recent years to further improve the transmission of multimedia sources. By learning from high-dimensional data patterns, DeepJSCC often exhibits a much more graceful degradation under worsening channel conditions compared with classic JSCC based on hand-crafted mappings, especially under finite code lengths. 
A pioneering work on DeepJSCC was proposed in \cite{8723589}, where the image pixel values are mapped to the complex-valued channel input symbols using an autoencoder structure. Subsequently, DeepJSCC schemes were extended to encompass text \cite{8-b2}, speech \cite{Yao2025} and video \cite{8-b5} sources, outperforming classic source-channel coding schemes.
Recently, the focus has shifted toward practical deployment challenges, such as adapting to diverse data and channel conditions while reducing computational overhead. The authors of \cite{8-b14} proposed an adaptive modulation and retransmission system that utilizes soft combining to improve the robustness of semantic tasks. In \cite{8-b16}, Shi et al. presented a rate-adaptive image DeepJSCC system by modeling semantic relations via a probabilistic graph over cross-attention mechanisms. 
In \cite{8-b17}, Shao et al. proposed an information bottleneck architecture with sparsity-inducing prior for adaptive feature pruning and utilized a dynamic neural network structure to accommodate different channel conditions.
Hu et al. \cite{10095672} presented a scalable DeepJSCC scheme for images based on the ranking of feature importance, improving the performance of intelligent tasks in the low SNR regime.
Wang et al. \cite{wang_resilient_2026} presented a class of reconfigurable neural networks that dynamically select the number of encoder and decoder blocks to address different channel conditions. Qin et al. \cite{MG-VIB} explored a novel class of information bottlenecks that improve scalability in broadcast settings.

Nonetheless, the performance of these methods degrades significantly when broadcasting to heterogeneous receivers. Consider, for instance, a base station utilizing a single encoder to broadcast information to edge devices. In such scenarios, users possess diverse communication and processing budgets. Prior DeepJSCC schemes are inherently ill-suited to this scenario due to the substantial cost of retransmissions and the need to conform to the weakest users in the wireless network. Moreover, they often require a scenario-specific encoder-decoder pair tailored to receivers' communication-computing capabilities, which cannot perform well across heterogeneous users without redesigning user-specific decoders. In summary, prior DeepJSCC schemes face the following fundamental challenges in broadcast channels:

\begin{itemize}
		{ \item \textit{Challenge 1: Underserving strong users.} The transmission bit length and processing complexity must be low enough to satisfy the most stringent individual communication-computing requirements, severely underserving other strong users upon broadcasting.}
		\item \textit{Challenge 2: High retransmission cost.} The transmitter must rebroadcast erroneous data to all users, implying that many users could potentially receive source information already obtained without improving performance.
		\item \textit{Challenge 3: Implicit rate-distortion-complexity tradeoff.} Edge devices experience diverse channel conditions and possess heterogeneous computing capabilities, thereby demanding a flexible rate-distortion-complexity tradeoff. However, due to the data-driven nature, most learning-based JSCC schemes cannot balance distortion, data rate, and processing complexity with theoretical insights. Instead, they depend solely on the static model architecture and trained parameters, which cannot generalize beyond the scenarios they were trained for.
\end{itemize}

In this paper, we address those challenges by designing a novel rate-adaptive JSCC scheme over the BIAWGN channel. Our goal is to design a source-channel coding framework with rateless properties (addressing Challenge 1 and 2) and a controllable rate-distortion-complexity tradeoff (Challenge 3), thereby enabling heterogeneous receivers, with different channel conditions and processing capabilities, to \textit{decode signals at varying quality levels and complexity from a single broadcast transmission}. Unfortunately, existing DeepJSCC schemes cannot achieve rateless transmission over wireless channels, whereas classic source-channel rateless coding schemes, which are often rooted in asymptotic analysis, cannot be optimized with deep learning modules and struggle to handle finite code lengths. To fill this research gap, our framework integrates {nonlinear transform coding (NTC)} with physical-layer Luby Transform (LT) codes. Specifically, the system learns a hyperprior-based side information to estimate bit likelihoods, which dynamically guide the encoding degree distributions and message selection probabilities of rateless symbols to preserve critical information. Since the transmitter generates a variable-length stream of coded symbols for broadcasting, each edge device can simply stop receiving symbols and adjust its belief propagation (BP) decoding iterations to match its specific channel condition and computing budget. The main contributions of this paper are summarized as follows:
\begin{itemize}
	\item We propose nonlinear transform rateless source-channel coding (NTRSCC), a variable-length JSCC framework based on Luby transformation (LT) codes that seamlessly incorporates learned nonlinear transforms and prior beliefs, providing a flexible and explicit tradeoff between distortion, transmission rate, and decoding complexity.
	
	\item We explore the parameter design of rateless codes when side information for each bit is available at the decoder, deriving unequal error protection strategies that guarantee successful decoder initiation while ensuring the informativeness of individual rateless symbols.
	
	\item We provide implementation and training details for NTRSCC, including approximation methods for rateless codes that allow its joint optimization with neural network modules in an end-to-end fashion.
	
	\item We perform numerical experiments across different image datasets and tasks to validate the effectiveness of our proposed system. Results show that NTRSCC could improve inference robustness, enable flexible performance scaling, and reduce total bandwidth under multicast settings.
	
\end{itemize}

The remainder of this paper is organized as follows. Section II formally describes the problem and reviews related works on rateless coding and nonlinear transform coding. Section III presents the schematic design of NTRSCC, and formulates its optimization from a variational perspective. Section IV investigates the parameter design of rateless codes when side information is present at the decoder. Section V presents numerical results and examples of our proposed method. Section VI concludes our paper. The key {notations} used in our paper are summarized in Table \ref{table0}. }

\begin{table}[tp]
	\centering
	\caption{{List of important notations in this paper}}
	\label{table0}
	\begin{small}
		\setlength{\extrarowheight}{3pt}
		\begin{tabular}{|c|c|}
			\hline
			Variable               & Definition                             \\ \hline
			$\mathbf{x}$           & Source image                           \\ \hline
			$\bf{y}, {\bf{b}}$ & Latent representation and quantized vector \\ \hline
			$\bf{z}, \bar{\bf{z}}, \tilde{\bf{z}}$ & Hyperprior latent variables and quantized / relaxed version \\ \hline
			$\mathbf{v}$           & Rateless coded symbols                  \\ \hline
			$\mathbf{u}$           & Channel input symbols                  \\ \hline
			$\hat{\mathbf{u}}_j$   & Channel output symbols at user $j$     \\ \hline
			$\hat{\mathbf{x}}_j$   & Reconstructed image at user $j$        \\ \hline
			$n_j$                  & Number of transmitted bits at user $j$ 			    \\ \hline
			$\sigma_j^2$           & Channel noise variance at user $j$     \\ \hline
			$\bm{\Omega}$          & Degree distribution of LT code 		\\ \hline
			$\bm{\rho}$                 & Message selection probability of LT code  \\ \hline
			$m_{i,o}^{(l)}$        & BP message from input $i$ to output $o$ \\ \hline
			$m_{o,i}^{(l)}$        & BP message from output $o$ to input $i$ \\ \hline
			$\mathcal{M}_i^{(l)}$  & Marginal log-likelihood of input $i$   \\ \hline
			$g_a, g_s$             & Nonlinear analysis and synthesis transforms \\ \hline
			$h_a, h_s$             & Hyperprior analysis and synthesis transforms \\ \hline
			$\bm{\eta}_j$          & Number of iterations at user $j$           \\ \hline
			$\bm{\gamma}_j$        & Scaling parameter for transmission length at user $j$           \\ \hline
			$K$                    & Number of users                        \\ \hline
			$N$                    & Maximum number of transmitted symbols per channel        \\ \hline
			$d_{max}$              & Maximum encoding degree \\ \hline
			$c$                    & Number of latent feature channels                        \\ \hline
			$k$                    & Length of quantized features per channel                        \\ \hline
		\end{tabular}
	\end{small}
\end{table}

{\section{Problem Description and {Preliminaries}}}

{\subsection{Problem Description}

Consider a central server which aims to send the features of image $\bf{x}$ to $K$ downlink edge users via broadcasting, where all users are interested in recovering the original image $\bf{x}$.
During transmission, the $j$-th user is connected to the server via an additional white Gaussian noise (AWGN) channel, receives channel symbol ${{\bf{\hat u}}}_j$, then infers a distorted version of the source image ${\bf{\hat x}}_j$. The above random variables constitute a {probabilistic model ${\bf x}\rightarrow{\bf u}\rightarrow{\bf \hat u}_j\rightarrow{\bf \hat x}_j$}.
We characterize the communication delay at each user by $n_j$ which is the number of received bits at the user, and characterize the processing delay at each user by ${\rm{Com}}{{\rm{p}}_j}$ which denotes computational complexity.
Due to the heterogeneous nature of the network, both the channel noise level, the acceptable transmission delay and the acceptable processing delay are all stochastic, e.g. $\sigma _j^2 \sim \pi \left( {{\sigma ^2}} \right)$, $T_j^{\left( 1 \right)} \sim \pi \left( {T_j^{\left( 1 \right)}} \right)$, $T_j^{\left( 2 \right)} \sim \pi \left( {T_j^{\left( 2 \right)}} \right)$.

Specifically, we aim to minimize the distortion at edge users $D({{{\bf{\hat x}}}_j};{\bf{x}})$, while satisfying the delay constraints of individual users. The problem could be formally defined as
\begin{align}\label{eq3}
	\min~ &\sum\limits_{j = 1}^K {D({{{\bf{\hat x}}}_j}; {\bf{x}})/K}, \notag\\
	\text{s.t. }\, &{n_j} \le T_j^{\left( 1 \right)}, {\text{Comp}_j} \le T_j^{\left( 2 \right)},~j = 1, 2, ..., K.
\end{align}
Here we define the maximum tolerable communication and processing delay as $T_1$ and $T_2$ respectively, both are constants determined by the transmission protocol.

Existing methods are unable to adjust transmission rate or complexity to accommodate heterogeneous delay budgets, therefore they will have to introduce enough slackness by drastically reducing the working number of bits $n_j$ and computational complexity ${\rm{Com}}{{\rm{p}}_j}$. This will greatly impact the distortion performance, which grows disproportionally with rate and complexity when retransmission is adopted.
Our solution to the above challenges is to \textit{combine learning-based source compression with physical layer rateless codes}, a class of rate-adaptive error correction codes that allows seamless rate adaptation, easy control over decoding complexity, and suitable for efficient broadcast.

\subsection{{Preliminaries 1:} Rateless Coding over BIAWGN Channels}

To help readers better understand the background of our work, we briefly introduce the system model for rateless channel coding.

Rateless codes are a class of error-control codes designed to approximate the ``digital fountain'' appeal proposed by Byers et. al \cite{5-b1}, which aims to eliminate the need for retransmissions while ensuring reliable data broadcast over an erasure channel. Consider a server transmitting a message consisting of a sequence of $k$ equal length packets to a {set} of clients. A client accepting any $n$ coded symbols from the channel could recover the original message with no error, where $n$ is a random variable slightly greater than $k$. The stream of coded symbols produced by the server in this solution is called a ``digital fountain''.
Since then, many good rateless channel codes have been proposed, e.g. LT codes \cite{5-b2} and raptor codes \cite{5-b3}, and their extensions to wireless channels have been extensively studied \cite{xu_optimization_2016} \cite{etesami_raptor_2006} \cite{sorensen_ripple_2014}.

For illustration, we present the encoding and decoding processes of an LT code over the BIAWGN channel. Assume that the transmitter has $k$ {message symbols} that needs to be transmitted. The LT encoder generates a stream of potentially infinite {coded symbols} by performing XOR operations on randomly chosen source symbols. We define the {degree} $d$ of a coded symbol as the number of source symbols that is used in the XOR operation. The degrees of output symbols follow a distribution ${\Omega}\left( d \right),~d = 1, 2, \ldots, k$, the design of which is crucial to the performance of LT codes. In certain cases involving unequal error protection (UEP), the encoder also utilizes a message selection probability $\rho (i),~i = 1, 2, \ldots, k$, which determines the selection probability of the $i$-th message symbol when forming each coded symbol.

The encoding process of LT codes can be summarized as follows. 
\begin{enumerate}
	\item Sample degree $d \sim {\Omega}\left( d \right)$, $d=1, 2, \ldots, k$.
	\item Select $d$ message symbols without repetition, with indices $i \sim \rho (i),~i = 1, 2, \ldots, k$.
	\item Perform XOR on the $d$ message symbols to generate a coded symbol.
\end{enumerate}

The {belief propagation (BP)} decoding of LT codes is done by iterative message passing between the input nodes and the output nodes. These messages incorporate the likelihood of the received symbols given the channel output and the prior beliefs from other nodes. Denote the message sent from input node $i$ to output node $o$ during the $l$-th round as $m_{i,o}^{\left( l \right)}$, the message from input node $i$ to output node $o$ during the $l$-th round as $m_{o, i}^{\left( l \right)}$, the incoming LLR from channel as ${{\hat v}}_o$, then the update rules for the $(l+1)$-th round are as follows:
\begin{align}
	\tanh \left( {{{m_{o,i}^{\left( l \right)}}}/{2}} \right) &= \tanh \left( {{{{{\hat v}_o}}}/{2}} \right) \cdot \prod\limits_{i' \ne i} {\tanh \left( {{{m_{i',o}^{\left( l \right)}}}/{2}} \right)}, \label{eq4} \\
	m_{i,o}^{\left( {l + 1} \right)} &= \sum\limits_{o' \ne o} {m_{o',i}^{\left( l \right)}}, \label{eq5}
\end{align}
where the product is over all input nodes adjacent to $o$ other than $i$, and the sum is over all output bits adjacent to $i$ other than $o$. The marginal log-likelihood of input nodes are calculated by $\mathcal{M}_i^{\left( l \right)} = \sum\limits_o {m_{o,i}^{\left( l \right)}}$.

The iteration process continues until convergence or a maximum number of iterations is reached, after which the decoder detaches itself from the channel and stops receiving more coded symbols.

It has been shown that, for a well-designed LT code under suitable channel noise levels, increasing the number of rateless symbols or the number of decoder iterations could improve error performance and enhance the quality of subsequent source decoding, thereby providing an explainable tradeoff between fidelity, rate and complexity \cite{etesami_raptor_2006}. 
In fact, employing rateless coding for JSCC has been an appealing idea. Clair et al. proposed a variable-length encoding scheme for binary source using rateless codes \cite{caire_universal_2004}. Bursalioglu et al. proposed a lossy JSCC scheme using raptor codes \cite{bursalioglu_lossy_2008}, and applied this scheme to deep space image communications \cite{bursalioglu_joint_2011}. Xu et al. explored distributed video coding using Raptor codes, as well as the design of an iterative soft decoder utilizing side information \cite{qian_xu_distributed_2007} \cite{qian_xu_wynerziv_2007}. 
However, those methods focus on accurate bit transmission and do not address nonlinear inference or critical feature extraction for edge users.

\subsection{{Preliminaries 2:} Nonlinear Transform Coding}

Nonlinear Transform Coding (NTC) was introduced to address the limitations of both conventional image codecs and early learning-based compression methods. Traditional schemes like JPEG rely on fixed transforms and manually designed side information, while early neural codecs often assume fully factorized priors that neglect spatial dependencies, leading to inefficient entropy coding. NTC bridges these gaps by jointly learning nonlinear transforms and a scale hyperprior that models correlations in latent representations. This enables more accurate entropy models and efficient quantization of semantic features, resulting in improved rate–distortion performance and adaptability across diverse image statistics.

The operation process of NTC is summarized as follows.
At the encoder, the source image ${\bf{x}}$ is first mapped into a latent representation ${\bf{y}}$ by a nonlinear analysis transform $g_a({\bf{x}};\phi_g)$. This transform extracts structural and semantic features from the image. The latent vector ${\bf{y}}$ is then quantized to a discrete sequence $\bar{\bf{y}}$.
To improve entropy modeling, a hyperprior analysis transform $h_a({\bf{y}};\phi_h)$ produces auxiliary latent variables ${\bf{z}}$, which are also quantized to $\bar{\bf{z}}$. These hyperlatents summarize the spatial distribution of scales in ${\bf{y}}$. A synthesis transform $h_s(\bar{\bf{z}};\theta_h)$ then estimates the scale parameters ${\bm{\sigma}}$, which parameterize the conditional distribution of each latent element $y_i$.

To enable gradient-based optimization despite the non-differentiable quantization step, NTC replaces hard quantization with additive uniform noise during training, defined as
\begin{equation}\label{eq7}
	{\bf{\tilde z}} = {\bf{z}} + o,~o \sim \mathcal{U}\left( { - 1/2, 1/2} \right).
\end{equation}
This relaxation yields a continuous latent representation that approximates the behavior of quantized variables, allowing the model to be trained end-to-end using standard backpropagation. 
Specifically, each latent variable $\tilde{y}_i$ is modeled as a Gaussian distribution with mean $\tilde{\mu}_i$ and variance $\tilde{\sigma}_i^2$, convolved with a uniform distribution to account for quantization. The parameters $(\tilde{\mu}, \tilde{\sigma})$ are predicted by a parametric synthesis transform $h_s$ applied to the auxiliary latent representation $\tilde{\mathbf{z}}$. This yields the conditional prior
\begin{equation}
	p(\tilde{\mathbf{y}}|\tilde{\mathbf{z}}) = \prod_i \left( \mathcal{N}(\tilde{y}_i|\tilde{\mu}_i, \tilde{\sigma}_i^2) * \mathcal{U}\left(-1/2, 1/2 \right) \right),
\end{equation}
The hyperprior $\tilde{\mathbf{z}}$ can be modeled using a non-parametric, fully factorized density model, given by
\begin{equation}
	p(\tilde{\mathbf{z}}|\bm{\psi}) = \prod_j \left( p_{z_j|\psi_j}(z_j|\psi_j) * \mathcal{U}\left(-1/2, 1/2\right) \right),
\end{equation}
where $\bm{\psi}$ are learnable parameters. Both $\tilde{\bf{y}}$ and $\tilde{\bf{z}}$ are transmitted along the communication channel, where the former comprises of the main bitstream, and the latter serves as the side information.

At the decoder, the side information $\tilde{\bf{z}}$ is first recovered and passed through $h_s$ to reconstruct the scale parameters ${\bm{\sigma}}$. These parameters define the probability mass functions used to recover the quantized latent sequence $\tilde{\bf{y}}$ via arithmetic decoding. Finally, the synthesis transform $g_s(\tilde{\bf{y}};\theta_g)$ reconstructs the image $\hat{\bf{x}}$.

The optimization of NTC can be interpreted as approximating the true posterior $p_{\tilde{\mathbf{y}}, \tilde{\mathbf{z}}|\mathbf{x}}$ with a parametric variational density $p_{\tilde{\mathbf{y}}, \tilde{\mathbf{z}}|\mathbf{x}}$, since the exact posterior is generally intractable. This approximation is achieved by minimizing the Kullback-Leibler (KL) divergence between the two distributions, averaged over the source distribution $p_{\mathbf{x}}$, expressed by
\begin{align}
	&\min_{\phi_g, \phi_h, \theta_g, \theta_h} \; \mathbb{E}_{p_(\mathbf{x})} \; D_{\mathrm{KL}} \left[ q(\tilde{\mathbf{y}}, \tilde{\mathbf{z}} | \mathbf{x}) \parallel p(\tilde{\mathbf{y}}, \tilde{\mathbf{z}} | \mathbf{x}) \right] \notag \\
	&= \min_{\phi_g, \phi_h, \theta_g, \theta_h} \; \mathbb{E}_{p_(\mathbf{x})} \Big[\log q (\tilde{\mathbf{y}}, \tilde{\mathbf{z}} | \mathbf{x}) \notag \\
	& - \log p (\tilde{\mathbf{y}} | \tilde{\mathbf{z}})  - \log p (\tilde{\mathbf{z}}) - \log p (\mathbf{x} | \tilde{\mathbf{y}})
	\Big] + \text{const}_1.
\end{align}
Here, $-\log p(\tilde{\mathbf{y}}|\tilde{\mathbf{z}})$ represents the transmission rate of the latent features based on the entropy model, $-\log p(\tilde{\mathbf{z}})$ is the transmission rate of side information, and $-\log p(\mathbf{x}|\tilde{\mathbf{y}})$ is the weighted distortion term, while $\log q(\tilde{\mathbf{y}},\tilde{\mathbf{z}}|\mathbf{x})$ could be technically dropped due to the uniformly-noised approximation. 

From the above formulation, we could see that NTC is suitable for the joint design with channel coding schemes, since it integrates quantization with a learned entropy model which could adaptively guide the length of the encoded bit stream to be transmitted through the channel.
Many prior works on JSCC have already employed NTC as the design framework, exploring applications in image \cite{8-b15}, video \cite{9953110} and more. 
{However, these methods, lacking the rateless property, still fail to address Challenge 1 and 2 under broadcast settings, and do not support complexity control as highlighted in Challenge 3.}
}
{\section{Schematic Design of NTRSCC}

\begin{figure*}[ht]
	\setlength{\belowcaptionskip}{-0.5cm}
	\vspace{-35pt}
	\centering
	\includegraphics[width=17.cm]{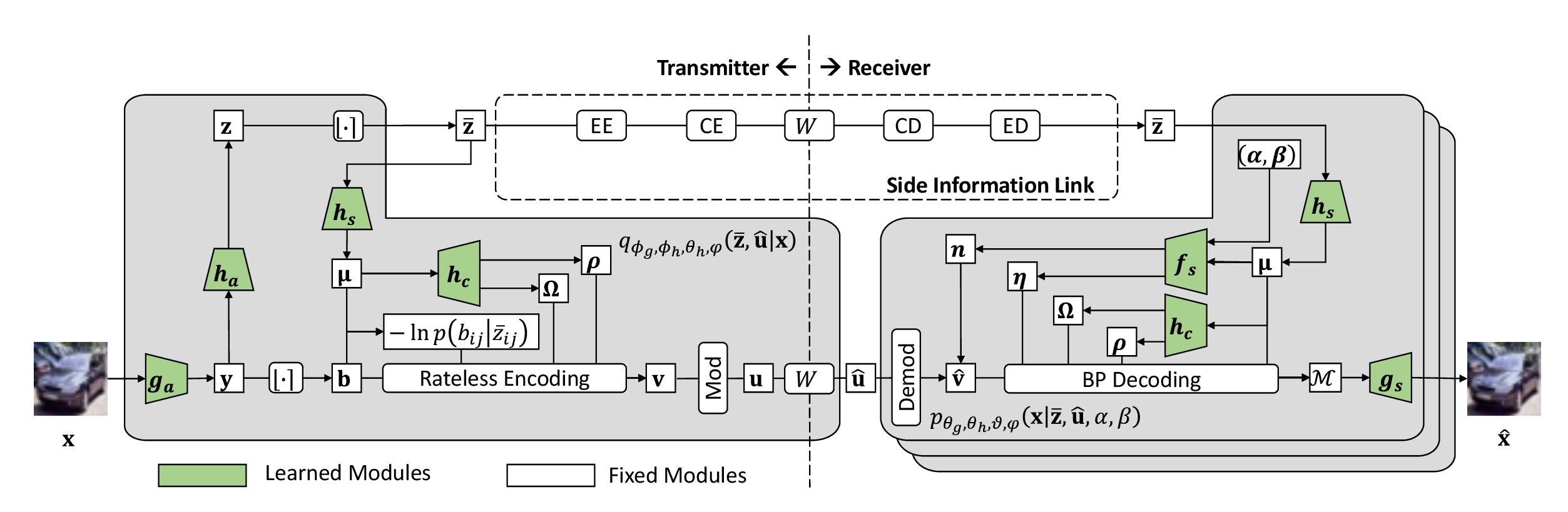}
	\caption{Schematic diagram for nonlinear transform  rateless source-channel coding. 
	At the transmitter, the analysis transform $g_a$ extracts semantic information ${\bf{y}}$ from source image, the analysis transform $h_a$ produces auxiliary variable $\bf{z}$, the synthesis transforms $h_s$ and $h_c$ predict bit likelihoods ${\bm{\mu}}$ and LT parameters ${\bm{\rho}}$, ${\bm{\Omega}}$ to generate rateless symbols $\bf{v}$.
	At the receiver, the synthesis transforms $h_s$ and $h_c$ predict bit likelihoods and LT parameters, the scaling function $f_s$ estimates the number of received symbols $\bf{n}$ and decoding iterations $\bm{\eta}$ to perform BP decoding on demodulator soft values ${\bf{\hat v}}$, the synthesis transform ${g_s}$ takes soft decoder outputs $\mathcal{M}$ and recovers image information.}
	\label{fig2}
\end{figure*}

{In this section, we focus on the overall structure of our proposed nonlinear transform rateless source-channel coding (NTRSCC). First, we provide a conceptual illustration of the working process for the NTRSCC. Then, we present a variational analysis of our proposed NTRSCC and identify its connections to our delay-constrained performance optimization problem. For notational simplicity, we focus on the transmission between the transmitter and a single user, and hence we will temporarily omit the user index in this part. Note that the analysis in this section could also be applied to the broadcasting case. For simplicity, we will focus on the case when binary phase-shift keying (BPSK) is used as the modulation scheme, which is a widely adopted assumption for many previous works on coding over BIAWGN channels \cite{xu_optimization_2016} \cite{etesami_raptor_2006} \cite{sharon_analysis_2006}.}

\subsection{System {Overview}}

The schematic diagram of NTRSCC is shown in Fig. \ref{fig2}. We elaborate on the system by the following components.
\subsubsection{Transmitter}
At the transmitter, the source image \({\bf{x}} \in \mathbb{R}^{H \times W \times C}\) is first mapped into a latent representation \({\bf{y}} \in [0, 1]^{k \times c}\) by a nonlinear analysis transform \(g_a({\bf{x}};\phi_g)\) that extracts semantic information from source image. ${\bf{y}}$ is then mapped to a bit sequence \({\bf{b}} \in \{0, 1\}^{k \times c}\) using straight-through estimation (STE) for quantization.
A hyperprior analysis transform \(h_a({\bf{y}};\phi_h)\) produces auxiliary latent variables \({\bf{z}} \in \mathbb{R}^{M}\).
Similar to the NTC formulation, we use a uniformly-noised representation ${{\bf{\tilde z}}}$ a continuous approximation {to quantization} during training, defined by Eq. (\ref{eq7}).
From ${{\bf{\tilde z}}}$, a synthesis transform \(h_s({{\bf{\tilde z}}};\theta_h)\) estimates the log-likelihood $\bm{\mu}$ of quantized bits, defined as
\begin{equation}
{\mu _{ij}} = \ln \frac{{p\left( {\left. {{b_{ij}} = 0} \right|{\bf{x}};{\theta _h}} \right)}}{{p\left( {\left. {{b_{ij}} = 1} \right|{\bf{x}};{\theta _h}} \right)}},\;i = 1,2, \ldots ,k,\;j = 1,2, \ldots ,c.
\end{equation}
{Here, we assume that the transformed bitstream $\bf{b}$ are independent. Using $\bm{\mu}$, we are able to devise strategies of rateless encoding, and provide side information to help rateless decoding.}
Additionally, {we design} a rateless coding parameter transform ${h_c}({\bm{\mu }};\varphi )$ {to predict} the encoding degree distributions ${\bf{\Omega }} \in {\left[ {0,1} \right]^{k \times c}}$ and the selection probability of each message bit ${\bm{\rho }} \in {\left[ {0,1} \right]^{k \times c}}$ for all $c$ encoded bit sequences.
{Moreover, we introduce a generation matrix} ${\bf{G}}$ {which is} sampled using ${\bm{\Omega }}$ and ${\bm{\rho }}$ in the way described in Section {II. B. We then obtain the quantized latent features}
\begin{equation} \label{eq13}
	{\bf{v}} = {\bf{Gb}},
\end{equation}
{which are} $c$ streams of rateless codewords with variable lengths:
\begin{equation} \label{eq14}
	{\bf{ v}} = \left[ {{{{\bf{ v}}}_1},{{{\bf{ v}}}_2} \ldots ,{{{\bf{ v}}}_c}} \right],{{{\bf{ v}}}_i} \in {\left\{ {0,1} \right\}^{{n_i}}},~i = 1,2, \ldots ,k,
\end{equation}
where $n_i$ is the number of rateless symbols produced for the $i-$th stream.
Furthermore, the side information ${{\bf{\bar z}}}$, {which is obtained by quantizing the auxiliary latent variables \({\bf{z}} \in \mathbb{R}^{M}\) to a string of integers ${{\bf{\bar z}}} \in \mathbb{Z}^{M}$,} is processed with entropy encoding (EE) and then encoded with a channel encoding (CE) at the transmitter during evaluation.

\subsubsection{Transmissions}
During transmissions, the PBSK modulator maps $\bf{v}$ into channel inputs $\bf{u}$, which are corrupted by additive white Gaussian noise (AWGN) to yield received symbols ${{\bf{\hat u}}}$. Under the AWGN channel, this process is expressed by
\begin{equation} \label{eq15}
	{\hat u_{ij}} = {u_{ij}} + w,~w \sim \mathcal{N}\left( {0,{\sigma ^2}} \right),~i = 1,2, \ldots ,k,~j = 1,2, \ldots ,c.
\end{equation}
The encoded side information is also modulated and transmitted over a wireless channel with the same noise level.

\subsubsection{Receiver}
The receiver collects noisy channel outputs ${\bf{\hat u}}$ from the $c$ different streams, which are demodulated to produce soft bit estimates \({\bf{\hat v}}\), which represent the log-likelihoods of the transmitted rateless codewords. Here, ${\bf{\hat v}}$ is defined as
\begin{equation} \label{eq16}
	{\bf{\hat v}} = \left[ {{{{\bf{\hat v}}}_1},{{{\bf{\hat v}}}_2} \ldots ,{{{\bf{\hat v}}}_c}} \right],{{{\bf{\hat v}}}_i} \in {\left\{ {0,1} \right\}^{{n_i}}},~i = 1,2, \ldots ,k,
\end{equation}
where $n_i$ is the number of rateless symbols acquired from the $i-$th stream.
The side information \({\bf{\bar z}}\) is recovered using entropy decoding (ED) and channel decoding (CD). A synthesis transform \(h_s({\bf{z}};\theta_h)\) then reconstructs the log-likelihood parameters \(\bm{\mu}\), which is utilized to guide belief propagation (BP) decoding in the following ways.
First, ${h_c}({\bm{\mu }};\varphi )$ produces the encoding degree distributions ${\bm{\Omega }}$ and the selection probability ${\bm{\rho }}$ at the decoder, which helps to construct the same bipartite graph structure with the encoder, provided with the same set of random seeds.
Second, \(\bm{\mu}\) is combined with the tradeoff parameters $\alpha$ and $\beta$ in a scaling function:
\begin{equation}\label{eq14}
    (\gamma, \eta) = {f_s}\left( {{\bm{\mu }},\alpha ,\beta ;\vartheta } \right),
\end{equation}
where $\gamma $ and $\eta$ are the scaling parameters of transmission bit length and iterations, respectively.
Finally, $\bm{\mu}$ is directly utilized in the BP decoding as prior beliefs on the likelihoods of the quantized features, to recover the soft log-likelihood values from the BP decoder output \({\bf{\hat y}} \in [0,1]^{k \times c}\). Specifically, the message update rules when \(\bm{\mu}\) is present are modified as
\begin{align}
	\tanh \left( {{{m_{o,i}^{\left( l \right)}}}/{2}} \right) &= \tanh \left( {{{{{\hat v}_o}}}/{2}} \right) \cdot \prod\limits_{i' \ne i} {\tanh \left( {{{m_{i',o}^{\left( l \right)}}}/{2}} \right)}, \label{eq10} \\
	m_{i,o}^{\left( {l + 1} \right)} &= \mu_{i} + \sum\limits_{o' \ne o} {m_{o',i}^{\left( l \right)}}, \label{eq11}
\end{align}
and the marginal log-likelihood for input nodes are calculated by $\mathcal{M}^{\left( l \right)} = \mu_{i} + \sum\limits_o {m_{o,i}^{\left( l \right)}}$,
where the number of iterations satisfy $l = 1,2, \ldots ,\left\lceil \eta  \right\rceil$. 
After obtaining the marginal likelihood, we do not perform hard decision to predict the latent bit sequence ${{\bf{\hat b}}}$, but instead calculate the probability
\begin{equation}
{p_1}\left( {{\bf{\hat b}}} \right) \buildrel \Delta \over = p\left( {\left. {{\bf{\hat b}} = 1} \right|{\bf{\bar z}},{\bf{\hat u}},{\bf{x}},\alpha ,\beta } \right) = 1 - {\rm{sigmoid}}\left( \mathcal{M} \right)
\end{equation}
This formulation allows the synthesis transform to operate on meaningful soft information by bounding the corrupted decoder outputs in a closed region, and enforces ${p_1}\left( {{\bf{\hat b}}} \right) \to {\bf{b}}$ when the sum-product algorithm approaches the true marginal. Moreover, it eliminates the need to infer the joint distribution $q({\bf{\bar z}},{\bf{\hat u}},{\bf{b}}\mid {\bf{x}})$ during end-to-end training, simplifying design.
Finally, we utilize a synthesis transform ${g_s}\left( {{p_1}\left( {{\bf{\hat b}}} \right);{\theta _s}} \right)$ to reconstruct the semantic output \({\bf{\hat s}} \in \mathbb{R}^{L}\).

\begin{Remark}
We assume that the entropy code and the channel code employed in the side information transmission link are efficient and capacity-achieving. Following the works of Slepian and Wolf \cite{slepian_noiseless_1973}, we intuitively define $n_j$ as proportional to the sum of log probabilities in the $j$-th feature channel, which is
\begin{equation} \label{eq21}
{n_j} =  - \frac{\gamma }{{{\rm{Cap}}\left( \sigma  \right)}} \cdot \sum\limits_{i = 1}^k {{{\log }_2}p\left( {\left. {{b_{ij}}} \right|{\bf{\bar z}}} \right)}.
\end{equation}
Here, $\gamma$ is a scaling parameter that enables lossy reconstruction based on the tradeoff parameters $\alpha$, $\beta$ of the current user, as well as the side information ${{\bf{\bar z}}}$. ${{\rm{Cap}}\left( \sigma  \right)}$ denotes the channel capacity measured in bits per channel symbol.
\end{Remark}

{\subsection{Variational Problem Formulation}}
\begin{figure*}[ht]
	\vspace{-35pt}
	\begin{align} \label{eq20}
		{\mathbb{E}_{p\left( {\bf{x}} \right)}}&{D_{KL}}\left[ {\left. {q({\bf{\tilde z}},{\bf{\hat u}}\mid {\bf{x}})} \right\|p({\bf{\tilde z}},{\bf{\hat u}}\mid {\bf{x}})} \right] \notag\\
		&= {\mathbb{E}_{p\left( {\bf{x}} \right)}}{\mathbb{E}_{q({\bf{\tilde z}},{\bf{\hat u}}\mid {\bf{x}})}}\left[ {\ln q({\bf{\hat u}}\mid {\bf{x}}) + \ln q({\bf{\tilde z}}\mid {\bf{x}}) - \ln p({\bf{\tilde z}}) - \ln p({\bf{\hat u}}\mid {\bf{\tilde z}}) - \ln p(\left. {\bf{x}} \right|{\bf{\tilde z}},{\bf{\hat u}}) + \ln p\left( {\bf{x}} \right)} \right] \notag\\
		&\mathop  = \limits^{\left( a \right)} {\mathbb{E}_{p\left( {\bf{x}} \right)}}{\mathbb{E}_{q({\bf{\tilde z}}\mid {\bf{x}})}}{\mathbb{E}_{q({\bf{\hat u}}\mid {\bf{x}})}}\left[ {\ln \frac{{q({\bf{\hat u}}\mid {\bf{x}})}}{{p({\bf{\hat u}}\mid {\bf{\tilde z}})}} + \ln q({\bf{\tilde z}}\mid {\bf{x}}) - \ln p({\bf{\tilde z}}) - \ln {\mathbb{E}_{p\left( {\left. {\mathcal{M}} \right|{\bf{\tilde z}},{\bf{\hat u}}} \right)}}p(\left. {\bf{x}} \right|{\mathcal{M}})} \right] + {\mathbb{E}_{p\left( {\bf{x}} \right)}}\ln p\left( {\bf{x}} \right) \notag\\
		&= {\mathbb{E}_{p\left( {\bf{x}} \right)}}\left\{ {{\mathbb{E}_{q({\bf{\tilde z}}\mid {\bf{x}})}}{D_{KL}}\left[ {\left. {q({\bf{\hat u}}\mid {\bf{x}})} \right\|p({\bf{\hat u}}\mid {\bf{\tilde z}})} \right] + {\mathbb{E}_{q({\bf{\tilde z}}\mid {\bf{x}})}}{\mathbb{E}_{q({\bf{\hat u}}\mid {\bf{x}})}}\left[ {\ln q({\bf{\tilde z}}\mid {\bf{x}}) - \ln p({\bf{\tilde z}}) - \ln {\mathbb{E}_{p\left( {\left. {\mathcal{M}} \right|{\bf{\tilde z}},{\bf{\hat u}}} \right)}}p(\left. {\bf{x}} \right|{\mathcal{M}})} \right]} \right\} + {\mathbb{E}_{p\left( {\bf{x}} \right)}}\ln p\left( {\bf{x}} \right) \notag\\
		&\mathop  = \limits^{\left( b \right)} {\mathbb{E}_{p\left( {\bf{x}} \right)}}\left\{ {{\mathbb{E}_{q({\bf{\tilde z}}\mid {\bf{x}})}}n{D_{KL}}\left[ {\left. {q(\hat u\mid {\bf{x}})} \right\|p(\hat u\mid {\bf{\tilde z}})} \right] + {\mathbb{E}_{q({\bf{\tilde z}}\mid {\bf{x}})}}{\mathbb{E}_{q({\bf{\hat u}}\mid {\bf{x}})}}\left[ {\ln q({\bf{\tilde z}}\mid {\bf{x}}) - \ln p({\bf{\tilde z}}) - \ln {\mathbb{E}_{p\left( {\left. {\mathcal{M}} \right|{\bf{\tilde z}},{\bf{\hat u}}} \right)}}p(\left. {\bf{x}} \right|{\mathcal{M}})} \right]} \right\} + {\mathbb{E}_{p\left( {\bf{x}} \right)}}\ln p\left( {\bf{x}} \right) \notag\\
		&\mathop  = \limits^{\left( c \right)} {\mathbb{E}_{p\left( {\bf{x}} \right)}}\left\{ {n{D_{KL}}\left[ {\left. {q(\hat u\mid {\bf{x}})} \right\|p(\hat u\mid {\bf{\tilde z}})} \right] - \ln p({\bf{\tilde z}}) - {\mathbb{E}_{q({\bf{\hat u}}\mid {\bf{x}})}}\ln {\mathbb{E}_{p\left( {\left. {\mathcal{M}} \right|{\bf{\tilde z}},{\bf{\hat u}}} \right)}}p(\left. {\bf{x}} \right|{\mathcal{M}})} \right\} + {\mathbb{E}_{p\left( {\bf{x}} \right)}}\ln p\left( {\bf{x}} \right) \notag\\
		&\mathop  \le \limits^{\left( d \right)} {\mathbb{E}_{p\left( {\bf{x}} \right)}}\left\{ {n{D_{KL}}\left[ {\left. {q(\hat u\mid {\bf{x}})} \right\|p(\hat u\mid {\bf{\tilde z}})} \right] - \ln p({\bf{\tilde z}}) - {\mathbb{E}_{q({\bf{\hat u}}\mid {\bf{x}})}}{\mathbb{E}_{p\left( {\left. {\mathcal{M}} \right|{\bf{\tilde z}},{\bf{\hat u}}} \right)}}\ln p(\left. {\bf{x}} \right|{\mathcal{M}})} \right\} + {\mathbb{E}_{p\left( {\bf{x}} \right)}}\ln p\left( {\bf{x}} \right) \notag\\
		&\mathop  \le \limits^{\left( e \right)} {\mathbb{E}_{p\left( {\bf{x}} \right)}}\left\{ {n{D_{KL}}\left[ {\left. {q(v\mid {\bf{x}})} \right\|p(v\mid {\bf{\tilde z}})} \right] - \ln p({\bf{\tilde z}}) - {\mathbb{E}_{q({\bf{\hat u}}\mid {\bf{x}})}}{\mathbb{E}_{p\left( {\left. {\mathcal{M}} \right|{\bf{\tilde z}},{\bf{\hat u}}} \right)}}\ln p(\left. {\bf{x}} \right|{\mathcal{M}})} \right\} + {\mathbb{E}_{p\left( {\bf{x}} \right)}}\ln p\left( {\bf{x}} \right)
		\tag{21}
	\end{align}
	\hrulefill
	\vspace{-10pt}
\end{figure*}
The optimization problem of NTRSCC can be formulated as a variational autoencoder (VAE) model: the probabilistic generative model is composed of demodulator, belief propagation decoding, coding parameter transform ${h_c}({\bm{\mu }};\varphi )$, and the synthesis transforms ${h_s}({\bf{z}};{\theta _h})$, ${g_s}({\bf{\hat y}};{\theta _g})$; the approximate inference model includes analysis transforms ${h_a}({\bf{y}};{\phi _h})$, ${g_a}({\bf{x}};{\phi _g})$, coding parameter transform ${h_c}({\bm{\mu }};\varphi )$, rateless encoder, modulator and the wireless channel. 
Similar to prior formulations, the inference model aims to construct a parametric variational density ${q_{{\phi _g},{\phi _h},{\theta _h},\varphi }}({\bf{\bar z}},{\bf{\hat u}}\mid {\bf{x}})$ that approximates the true posterior ${p_{{\phi _g},{\phi _h},{\theta _h},\varphi }}({\bf{\bar z}},{\bf{\hat u}}\mid {\bf{x}})$, which is intractable in practice. The training objective minimizes the KL divergence between these distributions over the source data \( p({\bf{x}}) \), which is given by
\begin{equation}
	\mathop {\min }\limits_{{\phi _g},{\phi _h},{\theta _h},{\theta _g},\varphi ,\vartheta } {\mathbb{E}_{p\left( {\bf{x}} \right)}}{D_{KL}}\left[ {\left. {q({\bf{\tilde z}},{\bf{\hat u}}\mid {\bf{x}})} \right\|p({\bf{\tilde z}},{\bf{\hat u}}\mid {\bf{x}})} \right]. \tag{20}
\end{equation}
Transforming this objective, we have Eq. (\ref{eq20}). Here, we define the total number of transmitted bits as 
$n = \sum\limits_{j = 1}^c {{n_j}}$.
Equality (a) comes from the fact that quantization noise $o$ and channel noise $w$ are independent, leading to 
\setcounter{equation}{21}
\begin{equation}
	q({\bf{\tilde z}},{\bf{\hat u}}\mid {\bf{x}}) = q({\bf{\tilde z}}\mid {\bf{x}}) \cdot q({\bf{\hat u}}\mid {\bf{x}}).
\end{equation}
Equality (b) comes from the i.i.d nature of both the Gaussian noise and the generation of rateless codewords.
Equality (c) comes from the definition of the uniformly-noised proxy variable ${{\bf{\tilde z}}}$, which is Eq. (\ref{eq7}).
Inequality (d) comes from Jensen inequality, i.e.
\begin{equation}
	\ln {\mathbb{E}_{p\left( {\left. {\mathcal{M}} \right|{\bf{\tilde z}},{\bf{\hat u}}} \right)}}p(\left. {\bf{x}} \right|{\mathcal{M}}) \ge {\mathbb{E}_{p\left( {\left. {\mathcal{M}} \right|{\bf{\tilde z}},{\bf{\hat u}}} \right)}}\ln p(\left. {\bf{x}} \right|{\mathcal{M}}).
\end{equation}
Inequality (e) comes from the data processing inequality.

Now we take a closer look on the variational objective derived in Eq. (\ref{eq20}).
The first term could be transformed as follows, given that ${q(v\mid {\bf{x}})}$ is exact:
\begin{align} \label{eq28}
	&n{\mathbb{E}_{p\left( {\bf{x}} \right)}}{D_{KL}}\left[ {\left. {q(v\mid {\bf{x}})} \right\|p(v\mid {\bf{\tilde z}})} \right] \notag\\
	&= n{\mathbb{E}_{p\left( {\bf{x}} \right)}}{\mathbb{E}_{q(v\mid {\bf{x}})}}{\mathbb{E}_{q({\bf{\tilde z}}\mid {\bf{x}})}}\left[ {\ln \frac{{q(v,{\bf{\tilde z}}\mid {\bf{x}})}}{{p(v\mid {\bf{\tilde z}})q({\bf{\tilde z}}\mid {\bf{x}})}}} \right] \notag\\
	&= n{\mathbb{E}_{p\left( {\bf{x}} \right)}}{\mathbb{E}_{q(v\mid {\bf{x}})}}{\mathbb{E}_{q({\bf{\tilde z}}\mid {\bf{x}})}}\left[ {\ln \frac{{q(v,{\bf{\tilde z}}\mid {\bf{x}})p\left( {\bf{x}} \right)}}{{p\left( {{\bf{\tilde z}}} \right)p(v\mid {\bf{\tilde z}})q({\bf{x}}\mid {\bf{\tilde z}})}}} \right] \notag\\
	&= n{\mathbb{E}_{p\left( {v,{\bf{\tilde z}},{\bf{x}}} \right)}}\left[ {\ln \frac{{q(\left. {v,{\bf{x}}} \right|{\bf{\tilde z}})}}{{p(v\mid {\bf{\tilde z}})q({\bf{x}}\mid {\bf{\tilde z}})}}} \right] \notag\\
	& = nI\left( {\left. {v;{\bf{x}}} \right|{\bf{\tilde z}}} \right) \ge nI\left( {\left. {\hat u;{\bf{x}}} \right|{\bf{\tilde z}}} \right),
\end{align}
which upper-bounds the amount of information that the noisy channel symbols ${{\bf{\hat u}}}$ could encapsulate about the source ${\bf{x}}$, when given side information ${{\bf{\tilde z}}}$. Moreover, since {${D_{KL}}\left[ {\left. {q(v\mid {\bf{x}})} \right\|p(v\mid {\bf{\tilde z}})} \right] \ge 0$}, minimizing the first term also reduces the number of rateless symbols $n$.
The second term corresponds to the rate of side information, encouraging compact representation of the prior probabilities and rateless coding parameters, similar to the NTC formulation.
The third term corresponds to the weighted distortion of the image being reconstructed from the BP decoder soft output. Specifically, for reconstruction tasks, this term is equivalent to
\begin{equation}
{{\mathbb{E}_{q({\bf{\hat u}}\mid {\bf{x}})}}{\mathbb{E}_{p\left( {\left. {\cal M} \right|{\bf{\tilde z}},{\bf{\hat u}}} \right)}}\ln p(\left. {\bf{x}} \right|{\cal M})}={\text{MSE}}\left( {{\bf{x}},{\bf{\hat x}}} \right) = \frac{\|{\bf{x}} - {\bf{\hat x}}\|^2}{{HWC}}.
\end{equation}
Here, $H$ denotes the image height, $W$ denotes the image width and $C$ is the number of channels.
Therefore, the objective can be seen as finding a tradeoff between the number of rateless symbols, the rate of side information, and weighted distortion of the reconstructed image, indicating a strong connection to our problem defined in Eq. (\ref{eq3}).

In order to optimize the variational objective, we must characterize ${{D_{KL}}\left[ {\left. {p(v\mid {\bf{x}})} \right\|p(v\mid {\bf{\tilde z}})} \right]}$ by investigating the encoding procedure of LT codes with non-uniform message symbol selection.

\section{{Optimization of NTRSCC}}

{In this section, to optimize the objective derived in the variational framework, we derive the theoretical principles governing the design of the rateless coding parameters. Specifically, we will design the encoding degree distribution $\Omega$ and the per-symbol selection probability $\rho$ by analyzing the requirements for successful BP initialization and maximizing per-symbol mutual information.} Since the rateless encoding is done separately for each feature channel with $k$ quantized bits, we will drop the indices for channels, and only consider the input node $i$ and the output node $o$ for the remainder of this section.

In classical information-theoretical analysis of rateless codes, the message bits are assumed to be all-zero for simplification, whose results could be extended to non-zero messages due to the symmetry of the BIAWGN channel. Therefore, we follow this convention by defining a dimension-wise flipping operation to the prior information $\bm{\mu}$, defined as
\begin{equation}
{{\tilde \mu }_{i}} = \left\{ {\begin{array}{*{20}{c}}
		{{\mu _{i}},~{b_{i}} = 0,}\\
		{ - {\mu _{i}},~{b_{i}} = 1,}
\end{array}} \right.
\end{equation}
where $i = 1,2, \ldots ,k$.
Furthermore, we assume that the prior estimation correctly reflects the true bit statistics, i.e. 
\begin{equation}
{\mu _{i}} \equiv \ln \frac{{p\left( {\left. {{b_{i}} = 0} \right|{\bf{x}}} \right)}}{{p\left( {\left. {{b_{i}} = 1} \right|{\bf{x}}} \right)}},
\end{equation}
which translates to
\begin{equation}
	p\left( {\left. {{{\tilde \mu }_{i}} < 0} \right|{\mu _{i}}} \right) = 1 - \sigma \left( {\left| {{\mu _{i}}} \right|} \right).
\end{equation}
This quantity reflects the probability that prior ${{\tilde \mu }_{i}}$ is on the 'wrong' side regarding to bit $b_{i}$.
Therefore, the distribution of the flipped priors ${{\bm{\tilde \mu }}}$ could be conditioned on ${\left| {{\mu _{i}}} \right|}$ alone, given by
\begin{equation}
{{\tilde \mu }_{ij}} = \left\{ {\begin{array}{*{20}{c}}
		{ - \left| {{\mu _{ij}}} \right|,p = 1 - {\rm{sigmoid}}\left( {\left| {{\mu _{ij}}} \right|} \right),}\\
		{\left| {{\mu _{ij}}} \right|,p = {\rm{sigmoid}}\left( {\left| {{\mu _{ij}}} \right|} \right),}
\end{array}} \right.
\end{equation}
where ${\rm{sigmoid}}\left( x \right) = 1/\left[ {1 + \exp \left( { - x} \right)} \right]$.

{To simplify notation, define
\begin{align}
	{V_o} & \buildrel \Delta \over = {\mathbb{E}_{p\left( {{{\hat v}_o}} \right)}}\left[ {\tanh \left( {{{{{\hat v}_o}}}/{2}} \right)} \right] \notag\\
	&=\frac{1}{{\sqrt {2\pi {\sigma ^2}} }}\int_{ - \infty }^\infty  {\tanh } \left( {{{\hat v}_o}/2} \right)\exp \left[ { - \frac{{{{(u - {\sigma ^2}/2)}^2}}}{{2{\sigma ^2}}}} \right]d{{\hat v}_o},
\end{align} 
where ${{\hat v}_o}$ is the demodulator output LLR under the all-zero assumption.
Define
\begin{align} \label{eq45}
	{U_i} &\buildrel \Delta \over = {\mathbb{E}_{p\left( {\left. {{{\tilde \mu }_i}} \right|{\mu _i}} \right)}}\left[ {\tanh \left( {{{\tilde \mu }_i}/2} \right)} \right]  \notag\\
	&= \left[ {2 \cdot {\rm{sigmoid}}\left( {\left| {{\mu _i}} \right|} \right) - 1} \right]\tanh \left( {\left| {{\mu _i}} \right|/2} \right).
\end{align}}
Define ${S}_d$ as the subset of input nodes connected to an output node of degree $d$. When the selection probability of each input node is $\rho_i$, the probability that a specific subset of message symbols are summed together by the encoding symbol is
\begin{equation}
p\left( {{{S}_d}} \right) = {{\prod\limits_{j \in {{ S}_d}} {{\rho _j}} } \mathord{\left/
		{\vphantom {{\prod\limits_{j \in {{S}_d}} {{\rho _j}} } {\sum\limits_{\left| {{{ S}_d}} \right| = d} {\prod\limits_{j \in {{S}_d}} {{\rho _j}} } }}} \right.
		\kern-\nulldelimiterspace} {\sum\limits_{\left| {{{ S}_d}} \right| = d} {\prod\limits_{j \in {{ S}_d}} {{\rho _j}} } }}.
\end{equation}
The output edge degree distribution $\omega(d)$ and input node edge degree distribution $\iota(d)$ are defined identically to those in \cite{etesami_raptor_2006}.

\subsection{Conditions for Successful Decoding Initiation}

To achieve flexible tradeoff between complexity and distortion, the LT codes should be designed in a way that each iteration at the BP decoder is helpful toward successful decoding \cite{xu_optimization_2016}, i.e. 
\begin{equation} \label{eq34}
	\mathbb{E}\left[ {m_{i,o}^{\left( {l + 1} \right)}} \right] > \mathbb{E}\left[ {m_{i,o}^{\left( l \right)}} \right].
\end{equation}
This constraint has been leveraged by many existing works through performance analysis of message-passing decoding for rateless codes under AWGN channels, utilizing methods like density evolution \cite{kharel_analysis_2018} and Gaussian approximation \cite{etesami_raptor_2006}. Unfortunately, those methods do not apply well in our joint optimization setting. The density evolution approach, while accurate, is computationally intensive and cannot be done over all feature channels and samples in training time. The Gaussian approximation method, while much simpler and fairly accurate, is not suitable for NTRSCC due to the Dirichlet distribution of prior $\bm{\mu}$, which violates the symmetric assumption underlying this approach. Those issues severely complicate the analysis of decoding performance for iterations beyond $l=0$.
In light of this dilemma, we instead take a hybrid approach by optimizing the $l>0$ iterations with a regularization loss term during end-to-end training, which is
\begin{equation}
{L_{inc}} = \frac{1}{{\left\lceil \eta  \right\rceil kn}}\sum\limits_{l = 1}^{\left\lceil \eta  \right\rceil } {\sum\limits_{i = 1}^k {\sum\limits_{o = 1}^n {\left[ {m_{i,o}^{\left( {l + 1} \right)} - m_{i,o}^{\left( l \right)}} \right]} } },
\end{equation}
and only analyze the necessary condition for successful decoder initialization.
{
The following theorem shows that the mean of messages passed from input to output nodes will always increase in the first iteration, which indicates successful decoder initialization on a global scale.
\begin{Theorem} \label{thm1}
	For non-trivial values of $\bm{\mu}$ with ${\left| {{\mu _{i}}} \right|} > 0$, $i=1, 2, \ldots, k$, there is always
	\begin{equation} \label{eq36}
		\mathbb{E}\left[ {m_{i,o}^{\left( {1} \right)}} \right] > \mathbb{E}\left[ {m_{i,o}^{\left( 0 \right)}} \right].
	\end{equation}
\end{Theorem}

\begin{proof}
	We set the initial message from input nodes to output nodes as $m_{i,o}^{\left( 0 \right)} = {{\tilde \mu }_i}$, which gives $\mathbb{E}\left[ {m_{i,o}^{\left( 0 \right)}} \right] = {\mathbb{E}_{p\left( {{\bm{\tilde \mu }}} \right)}}{{\tilde \mu }_{i'}}$.

	For an output node of degree $d_o$, taking the expectation on both sides of Eq. (\ref{eq10}) gives
	\begin{align}
		\mathbb{E}\left[ {\left. {\tanh \left( {{{m_{o,i}^{\left( 0 \right)}}}/{2}} \right)} \right|{d_o},{\bf{\tilde \mu }}} \right] = V_o \cdot {\mathbb{E}_{p\left( {{S_{{d_o}}}} \right)}}\prod\limits_{i' \in {S_{{d_o}}}\hfill\atop
			i' \ne i\hfill} {\tanh \left( {{{{{\tilde \mu }_{i'}}}}/{2}} \right)}.
	\end{align}
	Further averaging over the degree of output edge $d_o$ and message prior ${\bm{\tilde \mu }}$, we have
	\begin{align} \label{eq39}
		\mathbb{E}\left[ {\tanh \left( {{{m_{o,i}^{\left( 0 \right)}}}/{2}} \right)} \right] = {V_o} \cdot {\mathbb{E}_{\omega \left( {{d_o}} \right)}}{\mathbb{E}_{p\left( {{\bm{\tilde \mu }}} \right)}}{\mathbb{E}_{p\left( {{S_{{d_o}}}} \right)}}\prod\limits_{i' \in {S_{{d_o}}}\hfill\atop
			i' \ne i\hfill} {\tanh \left( {{{{{\tilde \mu }_{i'}}}}/{2}} \right)}.
	\end{align}
	Similarly, taking expectation over both sides of Eq. (\ref{eq11}) gives
	\begin{equation}\label{eq40}
		\mathbb{E}\left[ {m_{i,o}^{\left( 1 \right)}} \right] = {\mathbb{E}_{p\left( {{\bm{\tilde \mu }}} \right)}}{{\tilde \mu }_{i'}} + {\mathbb{E}_{\iota \left( {{d_i}} \right)}}\left( {{d_i} - 1} \right)
		\cdot\mathbb{E}\left[ {m_{o,i}^{\left( 0 \right)}} \right],
	\end{equation}
	where $d_i$ is the input edge degree. Note that Eq. (\ref{eq40}) only requires ${m_{o,i}^{\left( 0 \right)}}$ to be i.i.d instead of Gaussian. Combining the above euqations, Eq. (\ref{eq36}) transforms to
	\begin{equation}\label{eq41}
		\mathbb{E}\left[ {m_{i,o}^{\left( 1 \right)}} \right] - \mathbb{E}\left[ {m_{i,o}^{\left( 0 \right)}} \right] = {\mathbb{E}_{\iota \left( {{d_i}} \right)}}\left( {{d_i} - 1} \right)
		\cdot\mathbb{E}\left[ {m_{o,i}^{\left( 0 \right)}} \right]>0,
	\end{equation}
	which is equivalent to $\tanh \left(\mathbb{E} {\left[ {{{m_{o,i}^{\left( 0 \right)}}}/{2}} \right]} \right) > 0$.
	Noticing the function $\tanh \left( x \right)$ is strictly concave on $\left( {0, + \infty } \right)$, we only have to show that $\mathbb{E}\left[ {\tanh \left( {{{m_{o,i}^{\left( 0 \right)}}}/{2}} \right)} \right] > 0$.
	Further transforming Eq. (\ref{eq39}), we have
	\begin{align}
		&\mathbb{E}\left[ {\tanh \left( {{{m_{o,i}^{\left( 0 \right)}}}/{2}} \right)} \right]  \notag\\
		&= {V_o} \cdot {\mathbb{E}_{\omega \left( {{d_o}} \right)}}{\mathbb{E}_{p\left( {{S_{{d_o}}}} \right)}}\sum\limits_{{\bm{\tilde \mu }}} {\prod\limits_{i \in {S_{{d_o}}}} {p\left( {\left. {{{\tilde \mu }_i}} \right|{\mu _i}} \right)} \prod\limits_{i' \in {S_{{d_o}}}\hfill\atop
				i' \ne i\hfill} {\tanh \left( {{{{{\tilde \mu }_{i'}}}}/{2}} \right)} } \notag\\
		&= {V_o} \cdot {\mathbb{E}_{\omega \left( {{d_o}} \right)}}{\mathbb{E}_{p\left( {{S_{{d_o}}}} \right)}}\prod\limits_{i' \in {S_{{d_o}}}\hfill\atop
			i' \ne i\hfill} U_{i'}.
	\end{align}
	It is easy to verify that for all ${\left| {{\mu _{i'}}} \right|} > 0$, there is $U_{i'}>0$, therefore $\mathbb{E}\left[ {\tanh \left( {{{m_{o,i}^{\left( 0 \right)}}}/{2}} \right)} \right] > 0$, and Eq. (\ref{eq41}) always holds.
\end{proof}}
While Theorem \ref{thm1} indicates that BP decoding will always successfully start on a global scale, it does not guarantee good performance. In fact, for input nodes where prior values deviate from the true bit information, the local message update originating from that node could still lead to the wrong sign direction. Therefore, we must make sure the mean is sufficiently large to impose correct update direction on most messages. We introduce the first constraint as follows. 
{\begin{Remark}
Noticing that $U_j<1$, and that $\prod\limits_{i' \in {S_{{d_o}}}\hfill\atop
		i' \ne i\hfill} {{U_{i'}}}  > \prod\limits_{i \in {S_{{d_o}}}} {{U_i}},$
we can introduce the following constraint to facilitate successful initialization of BP decoding:
\begin{equation}
	{V_o} \cdot {\mathbb{E}_{\omega \left( {{d_o}} \right)}}{\mathbb{E}_{p\left( {{S_{{d_o}}}} \right)}}\prod\limits_{i \in {S_{{d_o}}}} {{U_i}}  > {\epsilon_1},
\end{equation}
with constant $\epsilon_1>0$.
\end{Remark}}

\subsection{Per-Symbol Mutual Information}
In this part, we {design the} conditional mutual information $I\left( {\left. {v;{\bf{x}}} \right|{\bf{\tilde z}}} \right)$, which illustrates the amount of information on image ${\bf{x}}$ conveyed by each rateless symbol ${v}$ with access to side information ${{\bf{\tilde z}}}$.
{The per-symbol mutual information serves as a regularizer that prevents ${\bm{\rho }}$ from collapsing to bits with high certainty to ensure successful BP initialization, which provides very little information regarding the source.}
Note that although we aim to minimize the total transmission rate across all rateless symbols, it is generally more ideal to achieve this by reducing the total number of transmitted symbols $n$, while maintaining the mutual information of each symbol, in order to achieve useful representation across a variety of transmission rates.
{
\begin{Theorem}
	The bit-wise mutual information is lower-bounded by 
	\begin{equation}\label{eq46}
		I\left( {\left. {v;{\bf{x}}} \right|{\bf{\tilde z}}} \right) \ge \frac{1}{2}{\mathbb{E}_{p\left( {\bf{x}} \right)}}{\left[ {{\mathbb{E}_{\omega \left( {{d_o}} \right)}}{\mathbb{E}_{p\left( {{S_{{d_o}}}} \right)}}\prod\limits_{j \in {S_{{d_o}}}} {{U_j}}  - 1} \right]^2}.
	\end{equation}
\end{Theorem}}
\begin{proof}
	From Eq. (\ref{eq28}), we have
	\begin{align}
		I\left( {\left. {v;{\bf{x}}} \right|{\bf{\tilde z}}} \right)
		&={\mathbb{E}_{p\left( {\bf{x}} \right)}}{D_{KL}}\left[ {\left. {q(v\mid {\bf{x}})} \right\|p(v\mid {\bf{\tilde z}})} \right]  \notag\\
		&= {\mathbb{E}_{p\left( {\left. {{\bf{\tilde z}}} \right|{\bf{x}}} \right)}}{\mathbb{E}_{p\left( {\bf{x}} \right)}}{\mathbb{E}_{q(v\mid {\bf{x}},{\bf{\tilde z}})}}\left[ {\ln \frac{{q(v\mid {\bf{x}},{\bf{\tilde z}})}}{{p(v\mid {\bf{\tilde z}})}}} \right] \notag\\
		&= {\mathbb{E}_{p\left( {\bf{x}} \right)}}{D_{KL}}\left[ {\left. {q(v\mid {\bf{x}},{\bf{\tilde z}})} \right\|p(v\mid {\bf{\tilde z}})} \right].
	\end{align}
	When the prior probability $\bm{\mu}$, the encoding degree distributions ${\bm{\Omega }}$ and the selection probability ${\bm{\rho }}$ are all inferred from the side information, the log-likelihood of rateless symbol $v$ is given by 
	\begin{equation}
		\tanh \left[ {{{LLR\left( {\left. {v} \right|{\bf{\tilde z}},{d_o},{S_{{d_o}}}} \right)}}/{2}} \right] = \prod\limits_{j \in {S_{{d_o}}}} {\tanh \left( {{{{\mu _j}}}/{2}} \right)},
	\end{equation}
	from which we can easily see it follows a Bernoulli distribution satisfying
	\begin{equation}
		p\left( {\left. {v = 0} \right|{\bf{\tilde z}}} \right) = \frac{1}{2}{\mathbb{E}_{\omega \left( {{d_o}} \right)}}{\mathbb{E}_{p\left( {{S_{{d_o}}}} \right)}}\left[ {1 + \prod\limits_{j \in {S_{{d_o}}}} {\tanh \left( {{{{\mu _j}}}/{2}} \right)} } \right].
	\end{equation}
	Similarly, when the source ${\bf{x}}$ is given, the encoding degree distributions ${\bm{\Omega }}$ and the selection probability ${\bm{\rho }}$ are inferred from the side information, the rateless symbol $v$ follows a Bernoulli distribution satisfying
	\begin{equation}
		p\left( {\left. {v = 0} \right|{\bf{x}},{\bf{\tilde z}}} \right) = \frac{1}{2}{\mathbb{E}_{\omega  \left( {{d_o}} \right)}}{\mathbb{E}_{p\left( {{S_{{d_o}}}} \right)}}\left[ {1 + \prod\limits_{j \in {S_{{d_o}}}} {\left( {1 - 2{b_j}} \right)} } \right].
	\end{equation}
	The calculation of both Bernoulli parameters require iterating over all subsets ${{S_{{d_o}}}}$ of shapes $d_o=1, 2, 3, \ldots, k$, making it impossible to calculate ${D_{KL}}\left[ {\left. {q(v\mid {\bf{x}},{\bf{\tilde z}})} \right\|p(v\mid {\bf{\tilde z}})} \right]$ during training time. Therefore, we instead investigate the properties of its lower bound.
	Utilizing the Pinkser bound for KL divergence between Bernoulli distributions, we obtain
	\begin{align}
		&{D_{KL}}\left[ {\left. {q(v\mid {\bf{x}},{\bf{\tilde z}})} \right\|p(v\mid {\bf{\tilde z}})} \right] \notag\\
		&\ge \frac{1}{2}{\left\{ {{\mathbb{E}_{\omega \left( {{d_o}} \right)}}{\mathbb{E}_{p\left( {{S_{{d_o}}}} \right)}}\left[ {\prod\limits_{j \in {S_{{d_o}}}} {\tanh \left( {{{{\mu _j}}}/{2}} \right) - \prod\limits_{j \in {S_{{d_o}}}} {\left( {1 - 2{b_j}} \right)} } } \right]} \right\}^2} \notag\\
		&= \frac{1}{2}{\left\{ {{\mathbb{E}_{\omega \left( {{d_o}} \right)}}{\mathbb{E}_{p\left( {{S_{{d_o}}}} \right)}}{\mathbb{E}_{p\left( {{\bf{\tilde \mu }}} \right)}}\left[ {\prod\limits_{j \in {S_{{d_o}}}} {\tanh \left( {{{{{\tilde \mu }_j}}}/{2}} \right) - 1} } \right]} \right\}^2} \notag\\
		&= \frac{1}{2}{\left[ {{\mathbb{E}_{\omega \left( {{d_o}} \right)}}{\mathbb{E}_{p\left( {{S_{{d_o}}}} \right)}}\prod\limits_{j \in {S_{{d_o}}}} {{U_j}}  - 1} \right]^2},
	\end{align}
	therefore Eq. (\ref{eq46}) holds.
\end{proof}

To make sure each rateless symbol contains enough useful information on the source, we must also constrain the bit-wise beyond a certain threshold.
{
\begin{Remark}
We can introduce the following constraint to facilitate useful representation for each rateless symbol:
\begin{equation} \label{eq52}
	\frac{1}{2}{\left[ {{\mathbb{E}_{\omega \left( {{d_o}} \right)}}{\mathbb{E}_{p\left( {{S_{{d_o}}}} \right)}}\prod\limits_{j \in {S_{{d_o}}}} {{U_j}}  - 1} \right]^2} > \epsilon_2,
\end{equation}
with constant $\epsilon_2>0$.
\end{Remark}}

\subsection{Unequal Error Protection Strategy}
\textbf{Determination of $\rho_j$}: To satisfy Eq. (\ref{eq46}) and Eq. (\ref{eq52}), a heuristic solution can be given by $\rho_j \propto \exp{(\lambda U_j)}$,
where $\lambda$ is a learnable parameter. Under this choice, the probability of a subset $S_d$ becomes
\begin{equation}
p_\lambda(S_d) = {\exp\!\left(\lambda \sum_{j \in S_d} U_j\right)}/{\sum_{|S_d|=d} \exp\!\left(\lambda \sum_{j \in S_d} U_j\right)}.
\end{equation}
This log–linear form ensures that the expectation
\begin{equation}
\Psi(\lambda) = \mathbb{E}_{\omega(d)} \, \mathbb{E}_{p_\lambda(S_d)} \prod_{j \in S_d} U_j 
\end{equation}
is a continuous function of $\lambda$, and increases monotonically with $\lambda$. 
Hence, by tuning $\lambda$ one can adjust $\Psi(\lambda)$ to satisfy both inequality constraints simultaneously. In practice, the parameters $\epsilon_1$ and $\epsilon_2$ are determined implicitly during end-to-end training.

\subsection{Decoding Complexity}
We define decoding complexity as the number of arithmetic operations that would occur in the BP decoder.
From Eq. (\ref{eq10}) and Eq. (\ref{eq11}), it can be seen that the belief propagation decoding involves passing update messages iteratively between sets of input and output nodes. Therefore, the computational complexity of BP decoding is determined jointly by the number of iterations and the structure of the decoding graph. 
{
	\begin{Theorem} \label{thm3}
	For an LT code with $k$ input symbols, $n$ output symbols, encoding degree distribution $\Omega(d)$ and iteration number ${\left\lceil \eta  \right\rceil }$, the total expected complexity is calculated by
	\begin{equation}\label{eq59}
		\text{Comp} = \left\lceil \eta  \right\rceil  \cdot \left[ {8n\sum\limits_{d = 1}^{d_{max}} {d \cdot \Omega \left( d \right)}  + 3n + k} \right].
	\end{equation}
\end{Theorem}}
\begin{proof}
	For an LT code of length $n$ with encoding degree distribution $\Omega(d)$, the expected total number of edges $E$ is determined solely by the encoding degree distribution, given by
	\begin{equation}
		{\mathbb{E}_{\Omega \left( d \right)}}E = n\sum\limits_{d = 1}^{d_{max}} {d \cdot \Omega \left( d \right)},
	\end{equation}
	where $d_{max}$ is the maximum encoding degree.
	
	\begin{table}[t] 
		\centering
		\caption{Arithmetic operations per BP decoding iteration}
		\label{tabel2}
		\begin{tabular}{|c|c|c|c|}
			\hline
			\textbf{Operation} & \textbf{Count} & \textbf{Operation} & \textbf{Count} \\
			\hline
			$\tanh(\cdot)$ & $E + n$ & $\tanh^{-1}(\cdot)$ & $E$ \\
			$/ $ & $2E + n$ & $\times$ & $2E + n$ \\
			$+$ & $E + k$ & $-$ & $E$ \\
			\hline
		\end{tabular}
		\vspace{-10pt}
	\end{table}
	
	For a graph with $E$ edges, $k$ input nodes and $n$ output nodes, the number of arithmetic operations for each iteration is summarized in Table \ref{tabel2}. From the table, the total expected number of operations per iteration is given by
	\begin{equation}
		{\mathbb{E}_{\Omega \left( d \right)}}\left[ {8E + 3n + k} \right] = 8n\sum\limits_{d = 1}^{d_{max}} {d \cdot \Omega \left( d \right)}  + 3n + k.
	\end{equation}
	Therefore, the total expected complexity for ${\left\lceil \eta  \right\rceil }$ iterations is calculated by
	\begin{align}
		\text{Comp} = \left\lceil \eta  \right\rceil  \cdot \left[ {8n\sum\limits_{d = 1}^{d_{max}} {d \cdot \Omega \left( d \right)}  + 3n + k} \right],
	\end{align}
	which proofs Theorem \ref{thm3}.
\end{proof}

Eq. (\ref{eq59}) reflects that increasing either the iteration count, the number of symbols, or the graph density directly raises the decoding cost. Therefore, we impose a maximum encoding degree constraint $d_{max}$.
Meanwhile, we adopt a known good degree distribution as the starting point of our optimization, following the analysis for stability condition in \cite{etesami_raptor_2006}, we set $\Omega(2)>1 / \ln(16)$.

\begin{figure*}[ht]
	\setlength{\belowcaptionskip}{-0.56cm}
	\vspace{-25pt}
	\centering
	\includegraphics[width=17cm]{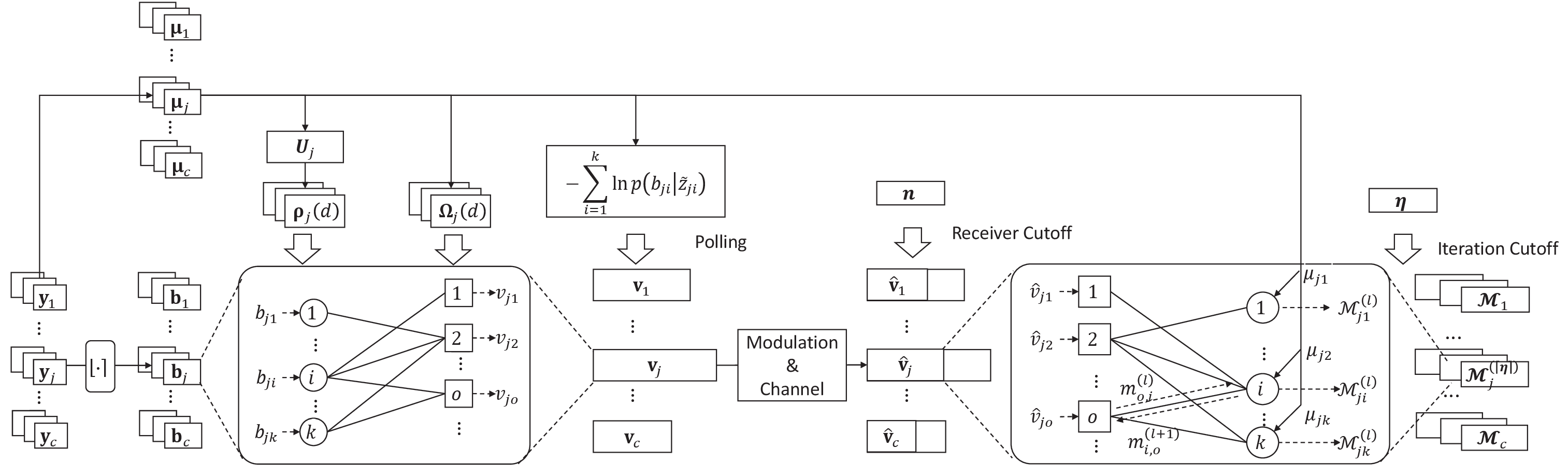}
	\caption{The LT encoding and decoding process in NTRSCC. Consider the j-th feature channel. The extracted latent features ${\bf{y}}_j$ are quantized to obtain bitstream ${\bf{b}}_j$. ${\bf{y}}_j$ also produces dimension-wise prior ${\bm{\mu}}_j$, which guides the prediction of degree distributions ${\bm{\Omega }}_j$ and ${\bm{\rho}}_j$. While the rateless symbols are produced indefinitely, the encoder employs a polling strategy based on the sum log-probability of each feature channel to generate channel symbols ${\bf{v}}_j$. The decoder stops receiving more rateless channel symbols when the total number of collected symbols reaches $\bf{n}$. The BP decoder processes demodulated LLR ${\bf{\hat v}}_j$ and the prior LLR ${\bm{\mu}}_j$, until the number of iterations reach ${\left\lceil \eta  \right\rceil }$, and the final outputs are the soft likelihoods $\mathcal{M}_j$.}
	\label{fig4}
\end{figure*}

\subsection{Detailed System Architecture}

The architecture of analysis and synthesis transforms are {detailed as follows.}
The analysis transform $g_a(\cdot; \phi_g)$ and synthesis transform $g_s(\cdot; \theta_g)$ are implemented as cascades of convolutional layers interleaved with Generalized Divisive Normalization (GDN) and its inverse (IGDN), respectively. Specifically, $g_a$ consists of three convolutional layers with kernel size $5 \times 5$ and stride 2, followed by GDN activations, and terminates with a sigmoid nonlinearity to bound the latent representation $\bf{y}$ between 0 and 1. The synthesis transform $g_s$ mirrors this structure using transposed convolution layers with matching kernel and stride parameters, interleaved with IGDN activations and concluding with a sigmoid output layer. The hyperprior analysis transform $h_a(\cdot; \phi_h)$ comprises three convolutional layers: a $3 \times 3$ layer with stride 1 followed by ReLU, then two $5 \times 5$ layers with stride 2 and ReLU activations. The hyperprior synthesis transform $h_s(\cdot; \theta_h)$ uses the same structure in reverse, employing transposed convolutions with ReLU activations to reconstruct the prior $\bm{\mu}$ that predicts the likelihhod of quantized bit sequence. 

The LT encoding and decoding processes in NTRSCC are outlined in Fig. \ref{fig4}. 
Consider the $j$-th channel of latent features, after obtaining the dimension-wise priors $\bm{\mu}_j$, the encoder calculates ${\bf{U}}_j$ using Eq. (\ref{eq45}). 
To approximately construct a generation matrix in a differentiable way, we propose the following procedures.
\begin{itemize}
	\item \textbf{Encoding degree sampling.}  
	Given a degree distribution ${{\bm{\Omega }}_j} = \left[ {{\Omega _{j1}},{\Omega _{j2}}, \ldots ,{\Omega _{j{d_{\max }}}}} \right]$, we apply the Gumbel-Softmax trick \cite{jang2017gumbelsoftmax} to obtain a one-hot vector representing the sampled degree:
	\begin{equation} \label{eq58}
		d \sim \mathrm{GumbelSoftmax}\!\left(\ln(\bm{\Omega}_j), \tau\right),
	\end{equation}
	where $\tau$ is the temperature parameter. With $\tau \to 0$, this approximates exact categorical sampling from $\omega$.
	
	\item \textbf{Encoding index sampling.}  
	Given degree $d$ and ${{\bf{U}}_j} = \left[ {{U _{j1}},{U _{j2}}, \ldots ,{U _{jk}}} \right]$, for each input symbol $i$, we perturb $(U_{ji})$ with independent Gumbel noise $g_{ji} \sim \mathrm{Gumbel}(0,1)$ and scale by a small temperature:
	\begin{equation}
		\ell_{ji} = (U_{ji} + g_{ji})/{\tau}.
	\end{equation}
	Sorting the logits $\ell_{ji}$ and taking the $d$-th largest sample $\ell (d)$ as threshold, we construct a relaxed mask
	\begin{equation} \label{eq60}
		{m_{ji}}(d) = {\rm{sigmoid}}\left[ {{\ell _{ji}} - \ell (d)} \right],
	\end{equation}
	which approximate binary indicators. The resulting mask selects approximately $d$ indices with probability proportional to $\prod_{i \in S_d} \rho_{ji}$.
\end{itemize}
The above steps are repeated until the maximum number of symbols is reached, yielding a relaxed generator matrix 
\begin{equation} \label{eq61}
	{\bf{G}} = \left[ {{m_{j1}},{m_{j2}}, \ldots ,{m_{jN}}} \right]
\end{equation}
suitable for encoding, while allowing gradient-based optimization.
After obtaining the $c$ streams of encoded rateless symbols ${\bf{v}}_1, {\bf{v}}_2, \ldots, {\bf{v}}_c$, the decoder applies random polling based on the sum log-probability of the message bits in each stream, in order to decide which stream to transmit in the next time slot. Specifically, the probability for transmitting a symbol from ${\bf{v}}_j$ is proportional to $- \sum_{i=1}^{k} \log_2 p(b_{ji} \mid \tilde{z}_{ji})$.

Each receiver keeps receiving more symbols from the wireless channel, until a total of $n = \frac{{\gamma ck}}{{{\rm{Cap}}\left( \sigma  \right)}}$ channel symbols are collected. For the $j$-th feature channel at the receiver, the average number of received symbols is hence given by Eq. (\ref{eq21}), where $\gamma$ comes from a learned transformation $\gamma = {f_{s1}}\left( {{\bm{\mu }},\alpha ,\beta ;\vartheta } \right)$. During training, we adopt a simplified approach and directly produce $n_j$ rateless symbols for the $j$-th feature channel during the encoding process.
The BP decoder iteratively updates the messages between input nodes and output nodes along the edges of the decoding graph, using Eq. (\ref{eq10}) and Eq. (\ref{eq11}), until the maximum number of iterations ${\left\lceil \eta  \right\rceil }$ is reached. Here, $\eta$ comes from a learned transformation $\eta = {f_{s2}}\left( {{\bm{\mu }},\alpha ,\beta ;\vartheta } \right)$. {The final outputs of the BP decoder are the soft likelihoods of each feature channel $[{\mathcal{M}}_1, {\mathcal{M}}_2, \ldots, {\mathcal{M}}_c]$.}

\subsection{Training of NTRSCC}
Recall the optimization problem defined in Eq. (\ref{eq3}), which is to reduce the average distortion across all connected users, while satisfying diverse rate and complexity constraints at individual users.
Therefore, we must optimize our system under different tradeoff levels between rate, complexity and distortion, while simultaneously enabling the traversal of those tradeoffs using a single encoder and decoder pair. 
To achieve this, we adopt different Lagrangian multiplier values $\beta_j$ for each user during training, and explicitly condition the number of transmitted features $N_j$ on $\beta_j$ using a scaling module.
Moreover, we apply an upper bound to the first term in Eq. (\ref{eq20}) to circumvent the intractable subset iterations, given by
\begin{align}
&{D_{KL}}\left[ {\left. {q({v_{ji}}\mid {\bf{x}},{\bf{\tilde z}})} \right\|p({v_{ji}}\mid {\bf{\tilde z}})} \right] \notag\\
&\le 4{\left\{ {{\mathbb{E}_{\omega \left( {{d_o}} \right)}}{\mathbb{E}_{p\left( {{S_{{d_o}}}} \right)}}\left[ {\prod\limits_{i \in {S_{{d_o}}}} {\tanh \left( {{\mu _{ji}}/2} \right) - \prod\limits_{i \in {S_{{d_o}}}} {\left( {1 - 2{b_{ji}}} \right)} } } \right]} \right\}^2} \notag\\
&= 4{\left\{ {{\mathbb{E}_{\omega \left( {{d_o}} \right)}}{\mathbb{E}_{p\left( {{S_{{d_o}}}} \right)}}{\mathbb{E}_{p\left( {{\bf{\tilde \mu }}} \right)}}\left[ {\prod\limits_{i \in {S_{{d_o}}}} {\tanh \left( {{{\tilde \mu }_{ji}}/2} \right) - 1} } \right]} \right\}^2} \notag\\
&\le 4{\left[ {{\mathbb{E}_{\omega \left( {{d_o}} \right)}}\prod\limits_{i \in {S_{\min {d_o}}}} {{U_{ji}}}  - 1} \right]^2} \buildrel \Delta \over = \hat D_{KL}^j
\end{align}
where ${{S_{\min {d_o}}}}$ denotes the subset containing the $d_0$ smallest elements.
During training, we first optimize the nonlinear transformation functions with the NTC loss function to improve reconstruction over quantized bits:
\begin{equation} \label{eq63}
{L_{NTC}} \buildrel \Delta \over = {\mathbb{E}_{p\left( {\bf{x}} \right)}}\left[ { - \log p({\bf{b}}|{\bf{\tilde z}}) - \log p({\bf{\tilde z}}) + {\rm{MSE}}\left( {{\bf{x}},{\bf{\hat x}}} \right)} \right],
\end{equation}
and optimize the rateless coding parameters using the decoder cross entropy to improve error performance:
\begin{equation} \label{eq64}
{L_{LT}} \buildrel \Delta \over = {\mathbb{E}_{\pi \left( {{\sigma ^2}} \right)}}\left[ { - \sum\limits_{j = 1}^c {\sum\limits_{i = 1}^k {\mathcal{M}_{ji}^{\left( {{\eta _{\max }}} \right)} \cdot \left( {1 - 2{b_{ji}}} \right)} }  + {L_{inc}}} \right],
\end{equation}
Combining the above analysis with the variational formulation in Section III part B, the final loss function is defined as
\begin{align} \label{eq62}
&{L_{NTRSCC}({\phi _h}, {\phi _g}, {\theta _h}, {{\theta _g}}, \varphi, \vartheta )} \buildrel \Delta \over = {\mathbb{E}_{p\left( {\bf{x}} \right)}}{\mathbb{E}_{p\left( \alpha  \right)}}{\mathbb{E}_{p\left( \beta  \right)}}{\mathbb{E}_{\pi \left( {{\sigma ^2}} \right)}} \notag\\
&\left\{ {{\rm{MSE}}\left( {{\bf{x}},{\bf{\hat x}}} \right) + \alpha \left[ {\sum\limits_{j = 1}^c {{n_j}\hat D_{KL}^j}  - \ln p({\bf{\tilde z}})} \right] + \beta \sum\limits_{j = 1}^c {{C_j}}  + {L_{inc}}} \right\}
\end{align}
Here, the first term represents the pixel-wise distortion of image, the second term is the total transmission rate weighted by tradeoff parameter $\alpha$, the third term is the BP decoding complexity multiplied by tradeoff parameter $\beta$, and the final term encourages the mean message increase between BP decoding iterations, to facilitate meaningful complexity-distortion tradeoff.

The training algorithm for multi-user case is summarized in Algorithm \ref{algorithm2}.
\RestyleAlgo{ruled}
\begin{algorithm}[t] 
	\caption{NTRSCC Training} \label{algorithm2}
	\textbf{Input:} Image dataset $\mathbf{x}$, tradeoff parameters $p(\alpha)$, $p(\beta)$, channel noise $p(\sigma^2)$, learning rate $\gamma$.\;
	\textbf{Parameters:} ${\phi _h}, {\phi _g}, {\theta _h}, {{\theta _g}}, \varphi, \vartheta $.\;
	
	\textbf{Phase I: NTC Pretraining}\;
	\For{$t=1,\ldots,T_1$}{
		Sample $\mathbf{x} \sim p(\mathbf{x})$.\;
		Compute $L_{NTC}$ with Eq. (\ref{eq63}).\;
		Update $\theta_g, \theta_h, {\phi _g}, {\phi _h}$.\;
	}
	
	\textbf{Phase II: Rateless Pretraining}\;
	\For{$t=1,\ldots,T_2$}{
		Sample $\sigma^2 \sim p(\sigma^2)$.\;
		Construct ${\bf{G}}$ using Eq. (\ref{eq58})-(\ref{eq61}).\;
		Compute ${{\bf{\hat v}}}$ using Eq. (\ref{eq13})-(\ref{eq16}).\;
		Perform BP decoding, compute $L_{LT}$ with Eq. (\ref{eq64}).\;
		Update $\varphi$.\;
	}
	
	\textbf{Phase III: Joint Optimization}\;
	\For{$t=1,\ldots,T_3$}{
		Sample $\sigma^2 \sim \pi(\sigma^2)$, $\mathbf{x} \sim p(\mathbf{x})$, $\alpha \sim p(\alpha)$, $\beta \sim p(\beta)$.\;
		Compute total loss $L_{NTRSCC}$ with Eq. (\ref{eq62}).\;
		Update ${\phi _h}, {\phi _g}, {\theta _h}, {{\theta _g}}, \varphi, \vartheta $.\;
	}
	
	\textbf{return} Trained parameters ${\phi _h}, {\phi _g}, {\theta _h}, {{\theta _g}}, \varphi, \vartheta $.\;
\end{algorithm}

\begin{Remark}
\textbf{(Rate-distortion-complexity tradeoff)} Different users can adjust $\alpha$ and $\beta$ on the fly according to their wireless channel conditions and computing capabilities. Higher $\alpha$ and $\beta$ correspond to stringent communication and computing budgets, respectively, whereas lower $\alpha$ and $\beta$ correspond to loose budgets. The decoder utilizes the scaling function in Eq. (\ref{eq14}). Note that the training process employs $\alpha$ and $\beta$ drawn from uniform distributions (Line 19 of Algorithm \ref{algorithm2}), and hence, the neural network learns generalized mapping that can be directly deployed on heterogeneous receivers that set $\alpha$ and $\beta$ according to their needs.
\end{Remark}

\begin{Remark}
\textbf{(Mapping between delay budgets and tradeoff parameters)} Due to the absence of closed-form expressions for the concrete rate-distortion-complexity bound, a practical way to determine suitable delay budgets is to construct a table of the achievable rate, distortion, and complexity as a function of $\alpha$ and $\beta$ during evaluation, and look up tradeoff parameters for a given distortion level during deployment.
\end{Remark}

\section{Results and Analysis}

In this section, we present numerical results of our proposed NTRSCC scheme to validate its channel-adaptive properties under dynamic channel conditions, as well as its capability to {reduce bandwidth consumption} under a multi-user setting.

\subsection{Experimental Setup}

\subsubsection{Datasets}

In this section, we select two datasets for image classification and reconstruction, namely the ImageNette (TinyImageNet) and CIFAR-10. The TinyImageNet dataset contains 100,000 images of size $64 \times 64$ across 200 object classes, with 500 training images and 50 validation images per class. The CIFAR-10 dataset consists of 60,000 $32 \times 32$ color images in 10 classes with 5,000 training images per class and 10,000 test images. During training and testing, the images are reshaped to $256 \times 256$.

\subsubsection{Comparison Schemes}
We employ several benchmarks defined as follows.
\begin{itemize}
	\item JPEG+LT: a separate source-channel coding scheme where the source coding is JPEG algorithm and the channel coding is LT codes with distribution used in Raptor codes \cite{etesami_raptor_2006}.
	\item NTC+LDPC: a separate source-channel coding scheme where the source coding is NTC and the channel coding is a LDPC code with code rate 0.67, column weight 3 and row weight 6. The number of decoder iterations is set to 10.
	\item NTRSCC (no UEP): identical to NTRSCC scheme but enforces uniform selection of message symbols.
\end{itemize}

\subsubsection{Evaluation Metrics} 
Consider an image which is $H$ pixels in height, $W$ pixels in width and has $C$ channels.
{\begin{itemize}
	\item \textit{Peak signal-to-noise ratio (PSNR):}
	PSNR is a pixel-wise distortion metric widely adopted in image evaluation literature. For an image where each channel as 256 greyscale levels, the PSNR is defined as
	\begin{equation}
		{\text{PSNR}}\left( {{\bf{x}},{\bf{\hat x}}} \right) = 10{\log _{10}}\left[ {255}^2 / {{{\text{MSE}}\left( {{\bf{x}},{\bf{\hat x}}} \right)}} \right],
	\end{equation}
	\item \textit{Bits per pixel (BPP):}
	Bits per pixel reflects how many bits are used to represent the source image, defined as
	\begin{equation}
	{\text{BPP}} = \left[ {n - \log p({\bf{\tilde z}})} \right]/HW.
	\end{equation}
	The numerator part represents the number of successfully received bits. The denominator part represents the number of pixels for the uncompressed source image.
	\item \textit{Operations per pixel (OPP):}
	Operations per pixel represents how many arithmetic operations are used during channel decoding, defined as
	\begin{equation}
		\text{OPP} = {\text{Comp}}/{H W}.
	\end{equation}
\end{itemize}
}

\subsubsection{Training Details}
In our experiments, all images are resized to $H=256$ and $W=256$ pixels. During training, we set $K=6$, where $\alpha _j\sim{\cal U}\left( {0,4} \right)$, $\beta _j\sim{\cal U}\left( {0,16} \right)$, $j = 1,2, \ldots ,K$
We set $c=64$, $N = 192$, $k=256$. The maximum encoding degree for rateless symbols is set to $d_{max}=16$. The number of training epochs is set to 20, the batch size is set to 64 and we employ the Adam optimizer with learning rate $3{\rm{e}} - 4$.

\subsection{Rate-Distortion Tradeoffs}

\begin{figure}[t]
	\setlength{\belowcaptionskip}{-0.3cm}
	\vspace{-10pt} 
	\centering
	\subfigure[PSNR versus BPP for CIFAR10 dataset.]
	{\includegraphics[width=6cm]{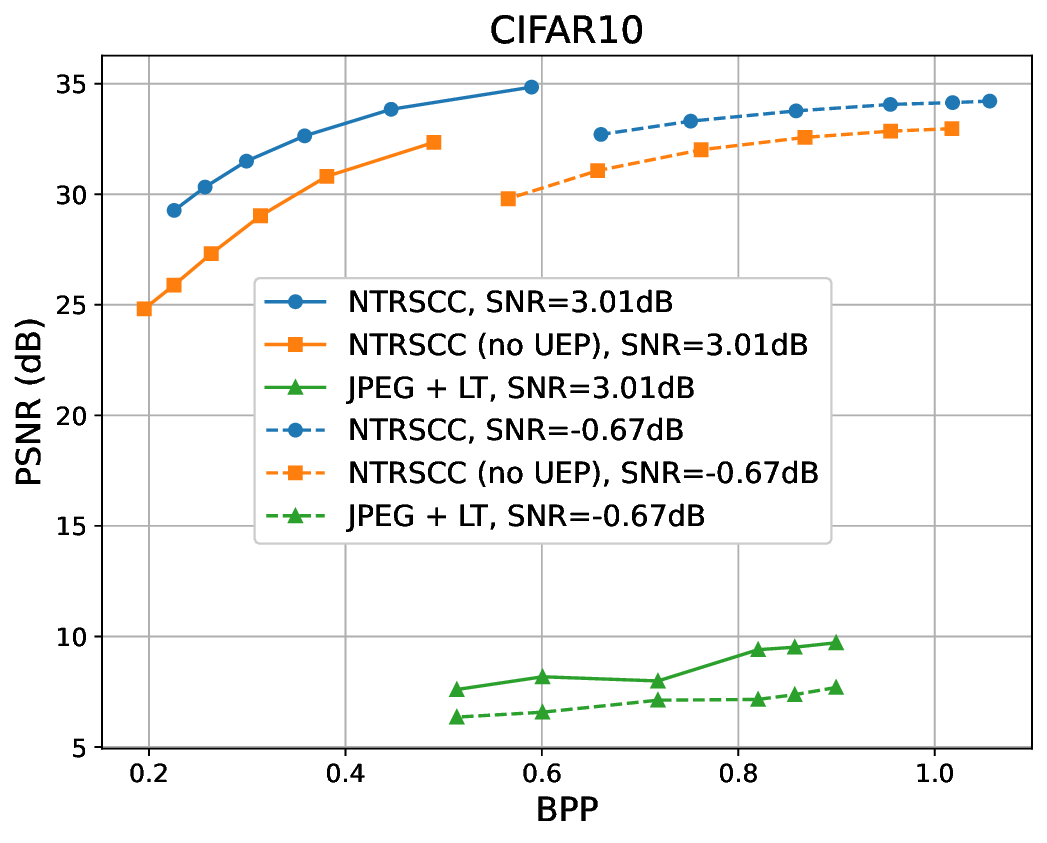}
		\label{fig6a}}
	\subfigure[PSNR versus BPP for ImageNette dataset.]
	{\includegraphics[width=6cm]{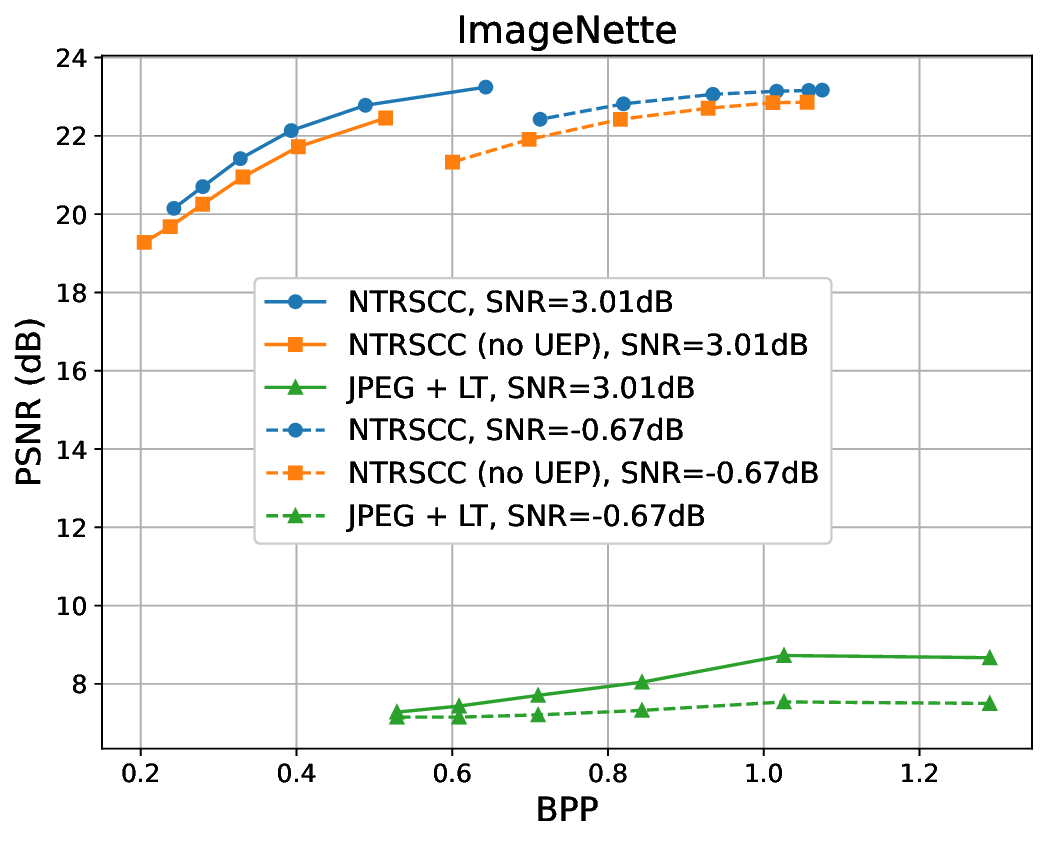}
		\label{fig6b}}
	\caption{Rate-distortion tradeoff comparisons.}
	\label{fig6}
\end{figure}
Fig. \ref{fig6} compare the achievable rate-distortion tradeoffs between NTRSCC, its uniform selection version, as well as the separation-based JPEG+LT. It could be seen that our proposed NTRSCC has the best PSNR performance when BPP is low, and is able to adjust BPP values in order to maintain task performance as channel worsens.
By comparison, JPEG+LT scheme has the worst performance across all data and channel conditions. This is due to the highly sensitive nature of traditional entropy coding to bit errors, where even slight corruptions at crucial segments could result in decoding failure of Huffman tables. Moreover, the JPEG algorithm itself does not provide bit-level prior information to the BP decoder during decoding, which results in higher channel decoding error probability compared to the other methods, further limiting its capabilities at low BPP and OPP.
Specifically, for the CIFAR10 dataset, when {SNR}=-0.67dB, {BPP}=0.71, the performance of NTRSCC is about 1.29dB higher than the uniform selection version, and 26.2dB higher than JPEG+LT.

\subsection{Complexity-Distortion Tradeoffs}
\begin{figure}[t]
	\setlength{\belowcaptionskip}{-0.3cm}
	\vspace{-10pt} 
	\centering
	\subfigure[PSNR versus OPP for CIFAR10 dataset.]
	{\includegraphics[width=6cm]{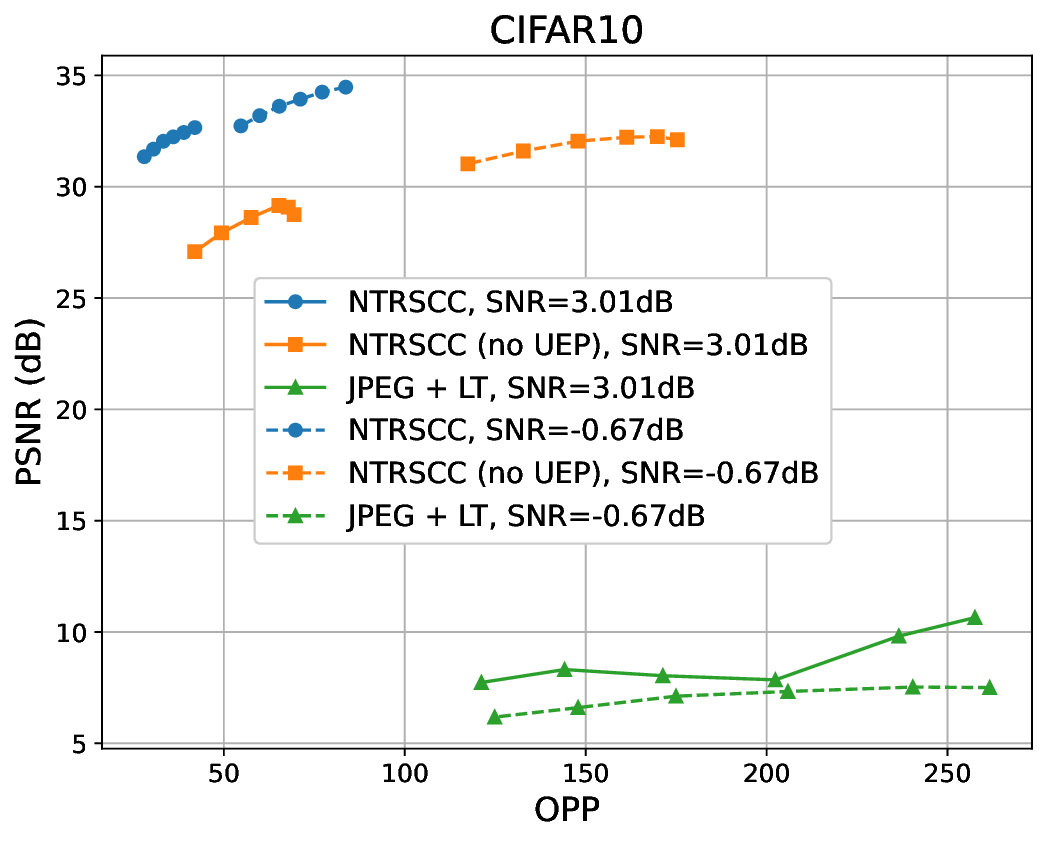}
		\label{fig7a}}
	\subfigure[PSNR versus OPP for ImageNette dataset.]
	{\includegraphics[width=6cm]{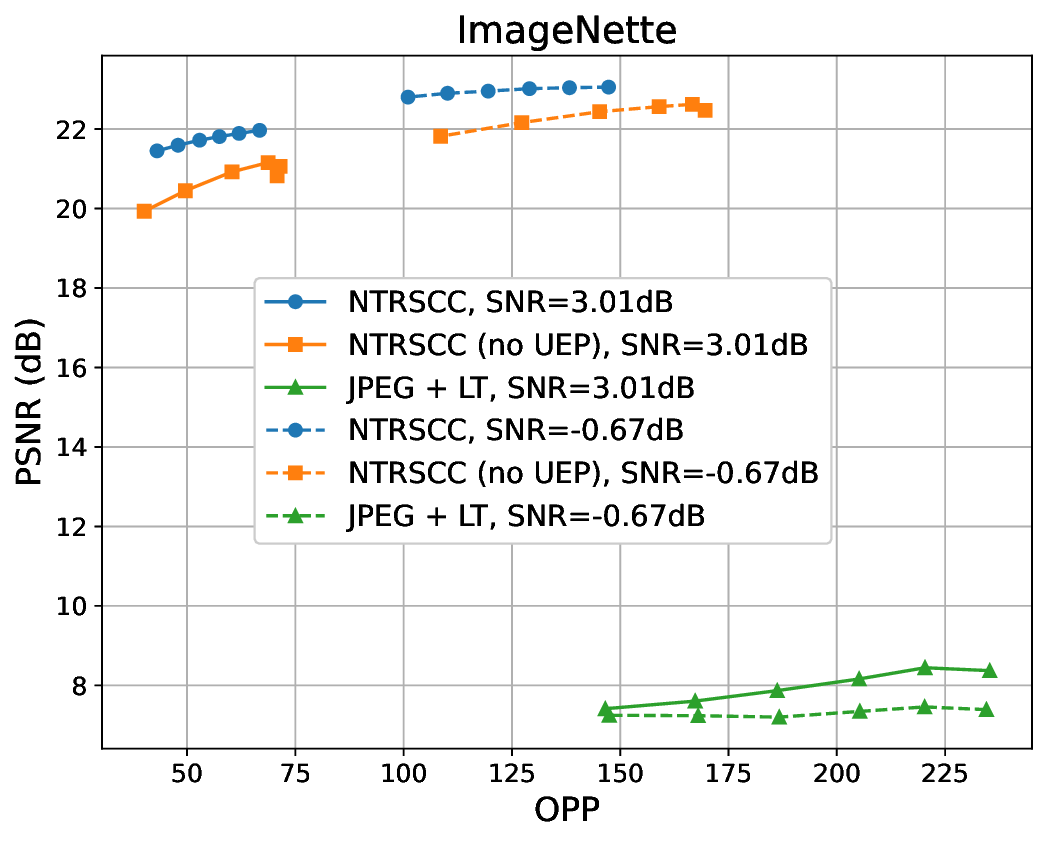}
		\label{fig7b}}
	\caption{Complexity-distortion tradeoff comparisons.}
	\label{fig7}
\end{figure}
Fig. \ref{fig6} compare the complexity-distortion tradeoffs between NTRSCC, and the benchmarks. Once again, NTRSCC has the best PSNR performance when BPP and OPP is low, and is able to adjust OPP in order to maintain the tradeoff level as channel worsens. The JPEG+LT method has the worst performance due to the suboptimality of separate designs.
Specifically, for the ImageNette dataset, when {SNR}=-0.67dB, {OPP}=145, the performance of NTRSCC is about 0.98dB higher than the uniform selection version, and 15.6dB higher than JPEG+LT.

\subsection{Efficient broadcast}

\begin{figure}[t]
	\setlength{\belowcaptionskip}{-0.25cm}
	\vspace{-10pt} 
	\centering
	\subfigure[Bits per pixel versus $K$ for CIFAR10 dataset.]
	{\includegraphics[width=6cm]{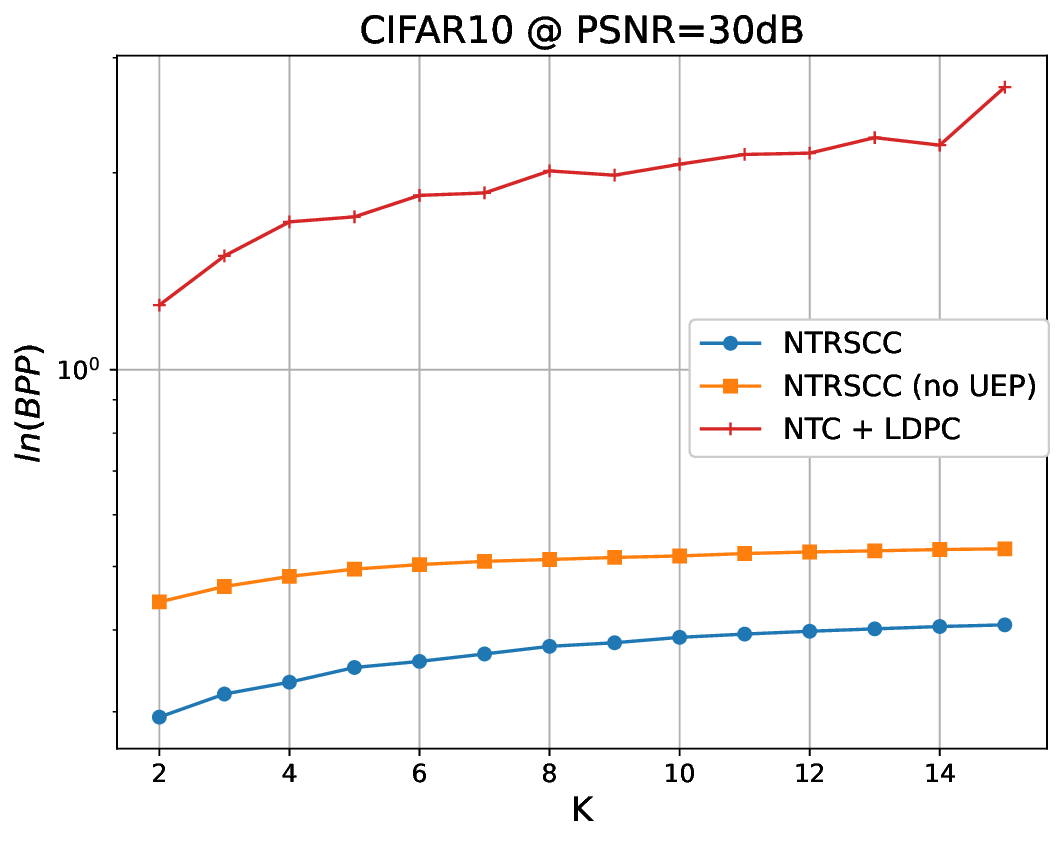}
		\label{fig8a}}
	\subfigure[Bits per pixel versus $K$ for ImageNette dataset.]
	{\includegraphics[width=6cm]{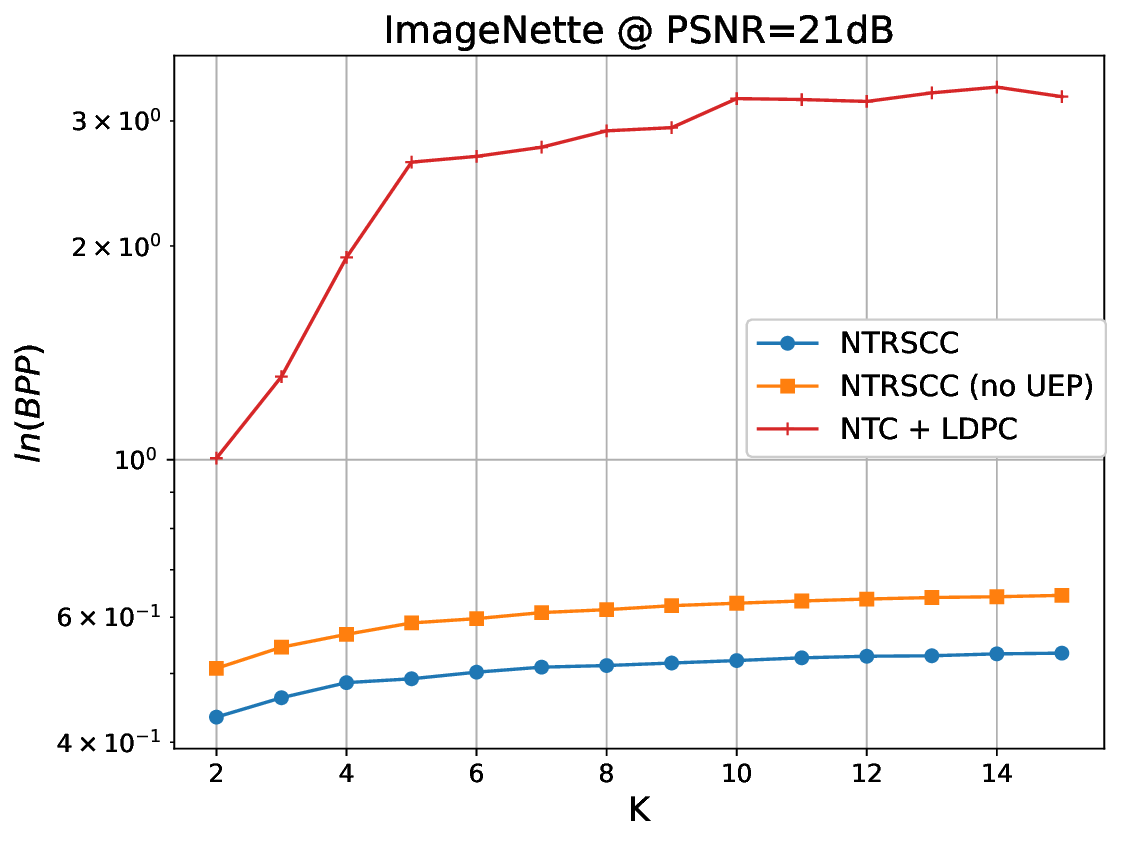}
		\label{fig8b}}
	\caption{broadcasting efficiency comparisons}
	\label{fig8}
\end{figure}

In this part we compare the number of coded symbols that has to be transmitted by the encoder in order to address all the connected users with respect to a certain performance threshold, under different numbers of users $K$. For CIFAR10 dataset, the threshold is set to \text{PSNR} = 30dB. For ImageNette dataset, the threshold is set to \text{PSNR} = 21dB. The channel noise for each user is sampled from distribution $\sigma _j^2\sim{\cal U}\left( {0,2} \right)$, $j = 1,2, \ldots ,K$. Furthermore, the NTC+LDPC employs a retransmission mechanism, that resend and averages feature channels with non-zero BER when the PSNR performance fails to reach the given threshold.
The results indicate that for both datasets, the proposed NTRSCC has the best broadcasting efficiency, which is attributed to its rateless nature as well as efficient UEP design.
The NTC+LDPC scheme has the worst broadcasting efficiency, since it is unable to address heterogeneous loss patterns and noise levels at all users.

\section{Conclusion}
In this work, we proposed NTRSCC for adaptive broadcast over heterogeneous wireless edge users.
First, we investigated the schematic framework of combining NTC with LT codes, formulating its design from a variational inference perspective. Then, we explored the design of rateless parameters for LT codes when bit-level side information is available at the decoder, accounting for informativeness, decoding success, and complexity. Finally, we provided approximations to the LT transmission process to allow end-to-end optimization, as well as detailed network architecture and training procedures for NTRSCC. The effectiveness of our method has been demonstrated through numerical experiments, allowing improved flexibility in balancing distortion, rate, and complexity in broadcasting. Future work could investigate the use of higher-order coding techniques and more efficient latent quantization methods.

\balance
\small
\bibliographystyle{IEEEtran}%
\bibliography{SecureFountain}

\end{document}